\documentclass[letterpaper,11pt]{article}
\usepackage[utf8]{inputenc}
\usepackage[namelimits]{amsmath} 
\usepackage{forest}
\usepackage{fullpage}
\usepackage{color}
\usepackage{amsfonts}
\usepackage[left=1.0in,right=1.0in,top=1.0in,bottom=1.0in]{geometry}
\definecolor{Blue}{rgb}{0.1,0.1,0.8}
\usepackage{hyperref}
\hypersetup{
    linktocpage=true,
    colorlinks=true,				
    linkcolor=Blue,				
    citecolor=Blue,				
    urlcolor=Blue,			
}
\usepackage{makecell}
\usepackage{tablefootnote}
\usepackage{algorithm}
\usepackage{algorithmic}
\usepackage{graphicx}
\usepackage{subcaption}
\usepackage{hyperref}
\usepackage{mathrsfs}
\usepackage[english]{babel}
\usepackage{amssymb}
\usepackage{multirow}
\usepackage{soul}
\usepackage{enumitem}
\usepackage{amsthm}
\usepackage{amsmath}
\usepackage{verbatim}
\usepackage{amssymb,tikz}
\usepackage{blindtext}
\usepackage{booktabs}
\usepackage{mathabx}
\newtheorem{theorem}{Theorem}

\newtheorem{lemma}{Lemma}

\theoremstyle{definition}

\newtheorem{definition}{Definition}
\newtheorem{claim}{Claim}
\newtheorem{fact}{Fact}
\newtheorem{example}{Example}

\newcommand{\argmax}{\mathop{\mathrm{arg\,max}}}
\newcommand{\argmin}{\mathop{\mathrm{arg\,min}}}

\newcommand{\X}{\mathbf{X}}

\newcommand{\Pm}{\mathbf{P}}
\newcommand{\bSigma}{\mathbf{\Sigma}}
\newcommand{\Lam}{\mathbf{\Lambda}}
\newcommand{\x}{X}

\newcommand{\W}{\mathbf{W}}
\newcommand{\tr}{\mathrm{tr}}

\title{Differentially Private Covariance Revisited}
\author{
  Wei Dong, Yuting Liang, Ke Yi\\
  \texttt{\{wdongac,yliangbs,yike\}@cse.ust.hk}\\
  Department of Computer Science and Engineering\\
  Hong Kong University of Science and Technology
}

\begin{document}
\maketitle 

\begin{abstract}
In this paper, we present two new algorithms for covariance estimation under concentrated differential privacy (zCDP).  The first algorithm achieves a Frobenius error of $\tilde{O}(d^{1/4}\sqrt{\tr}/\sqrt{n} + \sqrt{d}/n)$, where $\tr$ is the trace of the covariance matrix.  By taking $\tr=1$, this also implies a worst-case error bound of $\tilde{O}(d^{1/4}/\sqrt{n})$, which improves the standard Gaussian mechanism's $\tilde{O}(d/n)$ for the regime $d>\widetilde{\Omega}(n^{2/3})$.  Our second algorithm offers a tail-sensitive bound that could be much better on skewed data.  The corresponding algorithms are also simple and efficient. Experimental results show that they offer significant improvements over prior work.
\end{abstract}

\section{Introduction}
Consider a dataset represented by a matrix $\X\in\mathbb{R}^{d\times n}$, where each column $X_i, i=1,\dots,n$ corresponds to an individual's information.  As standard in the literature, we assume that all the $\x_i$'s live in $\mathcal{B}_d$, the $d$-dimensional $\ell_2$-unit ball  centered at the origin. In this paper, we revisit the  problem of estimating the (empirical) covariance matrix $\mathbf{\Sigma}(\X):=\frac{1}{n}\sum_i \x_i\x_i^T = \frac{1}{n} \X \X^T$ under differential privacy (DP), a fundamental problem in high-dimensional data analytics and machine learning that requires little motivation. We often write $\mathbf{\Sigma}(\X)$ as $\mathbf{\Sigma}$ when the context is clear. As with most prior work, we use the Frobenius norm $\|\widetilde{\mathbf{\Sigma}} - \mathbf{\Sigma}\|_F$ to measure the error of the estimated covariance $\widetilde{\mathbf{\Sigma}}$.  
To better focus, in the introduction we state all results under \textit{concentrated different privacy (zCDP)} \cite{bun2016concentrated}; extensions of our results to pure-DP are given in 
Section \ref{sec:pure_dp}.  

\subsection{A Trace-sensitive Algorithm}
\label{sec:traceintro}
For any symmetric matrix $\mathbf{A}$, we use $\mathbf{P}[\mathbf{A}]$ and $\mathbf{\Lambda}[\mathbf{A}]$ to denote its matrices of eigenvectors and eigenvalues, respectively, such that $\mathbf{A}=\mathbf{P}[\mathbf{A}]\mathbf{\Lambda}[\mathbf{A}]\mathbf{P}[\mathbf{A}]^T$; we use $\lambda_i[\mathbf{A}]$ to denote its $i$th largest eigenvalue.  When $\mathbf{A}=\mathbf{\Sigma}=\mathbf{\Sigma}(\mathbf{X})$, we simply write $\Pm =\Pm[\mathbf{\Sigma}], \mathbf{\Lambda}=\mathbf{\Lambda}[\mathbf{\Sigma}], \lambda_i=\lambda_i[\mathbf{\Sigma}]$, so that $\mathbf{\Lambda} = \mathrm{diag}(\lambda_1,\dotsb,\lambda_d)$ and $\mathbf{\Sigma} = \mathbf{P}\mathbf{\Lambda}\mathbf{P}^T$.  Let $\Pm = [P_1\; P_2\; \dotsb \; P_d]$, where $P_i$ is the orthonormal basis vector corresponding to $\lambda_i$. Rudimentary linear algebra yields $\lambda_k = \frac{1}{n} \sum_i (P_k^T X_i)^2$ for $1\leq k \leq d$ and $||X_i||_2^2 = \sum_k (P_k^T X_i)^2$ for $1\leq i \leq n$. Thus, it follows that 
\begin{equation*}
\label{eq:tr}
\mathrm{tr}[\mathbf{\Sigma}] = \mathrm{tr}[\mathbf{\Lambda}] = \sum_k \lambda_k = \sum_k \frac{1}{n} \sum_i (P_k^T X_i)^2 = \frac{1}{n}\sum_i\sum_k(P_k^T X_i)^2 = \frac{1}{n}\sum_i||X_i||_2^2.
\end{equation*}
That is, $0\le \mathrm{tr}[\mathbf{\Lambda}]\le 1$ is the average $\ell_2$ norm (squared) of the $\x_i$'s, and we simply write it as $\mathrm{tr}$.  

Recall that it is assumed that all the $X_i$'s live in $\mathcal{B}_d$.  In practice, this is enforced by assuming an upper bound $B$ on the norms and scaling down all $X_i$ by $B$. As one often uses a conservatively large $B$, typical values of $\tr$ can be much smaller than $1$, so a trace-sensitive algorithm would be more desirable.  Indeed, Amin et al.~\cite{amin2019differentially} take this approach, describing an algorithm with error\footnote{ We use the $\tilde{O}$ notation to suppress the dependency on the privacy parameters and all polylogarithmic factors. We use $e$ as the base of $\log$ (unless stated otherwise) and define $\log(x) = 1$ for any $x\leq e$.} $\tilde{O}(d^{3/4}\sqrt{\tr}/\sqrt{n} + \sqrt{d}/n)$ under zCDP\footnote{Their paper states the error bound under pure-DP and for estimating $\X\X^T$ (i.e., without normalization by $1/n$); we show how this bound is derived from their result in Appendix \ref{sec:approximate_DP_nips_2019}.}.  Note that the $\sqrt{d}/n$ term inherits from mean estimation and the first term is the ``extra'' difficulty for covariance estimation. In this paper, we improve this term to $d^{1/4}\sqrt{\tr}/\sqrt{n}$ (we have a similar, albeit lesser, improvement under pure-DP; see Section \ref{sec:pure_dp}).  Our algorithm is very simple: We first estimate $\mathbf{\Lambda}$ using the Gaussian mechanism (this is the same as in \cite{amin2019differentially,kapralov2013differentially}), then we estimate $\mathbf{P}$ by doing an eigendecomposition of $\bSigma$ masked with Gaussian noise.  Intuitively, we obtain a $\sqrt{d}$-factor improvement over the iterative methods of \cite{amin2019differentially,kapralov2013differentially}, because we can obtain all eigenvectors from one noisy $\bSigma$, while the iterative methods must allocate the privacy budget to all $d$ eigenvectors.  Our algorithm is also more efficient, performing just two eigendecompositions and one matrix multiplication, whereas the algorithm in \cite{amin2019differentially,kapralov2013differentially} needs $O(d)$ such operations.

\paragraph{Implication to worst-case bounds.}
Covariance matrix has also been studied in the traditional worst-case setting, i.e., the bound should only depend on $d$ and $n$.  Dwork et al.~\cite{dwork2014analyze} show that the $\ell_2$-sensitivity of $\mathbf{\Sigma}$, i.e., $\max_{\X \sim \X'} \|\bSigma(\X)-\bSigma(\X')\|_F$ where $\X\sim \X'$ denotes two neighboring datasets that differ by one column, is $O(1/n)$.  Thus, the standard \textit{Gaussian mechanism} achieves an error of $\tilde{O}(d/n)$ by adding an independent Gaussian noise of scale $\tilde{O}(1/n)$ to each of the $d^2$ entries of $\bSigma$.  By taking $\tr=1$, our trace-sensitive bound degenerates into $\tilde{O}(d^{1/4}/\sqrt{n})$.  Note that the $\sqrt{d}/n$ term is dominated by $d^{1/4}/\sqrt{n}$ for $d<\tilde{O}(n^2)$, which is the parameter regime that allows non-trivial utility (i.e., the error is less than $1$).   

To better understand the situation, it is instructive to compare covariance estimation with mean estimation (where data are also drawn from the $\ell_2$ unit ball and the error is measured in $\ell_2$ norm), as the hardness of covariance estimation lies between $d$-dimensional mean estimation (only estimating the diagonal entries of $\bSigma$) and $d^2$-dimensional mean estimation (treating $\bSigma$ as a $d^2$-dimensional vector). 
This observation implies a lower bound $\widetilde{\Omega}(\sqrt{d}/n)$ following from the same lower bound for mean estimation \cite{KamathLSU19}\footnote{This paper proves the lower bound under the statistical setting; in Appendix~\ref{sec:lower_bound_statistical_to_empirical}, we show how it implies the claimed lower bound under the empirical setting.}, and an upper bound $\tilde{O}(d/n)$ attained by the Gaussian mechanism. For $d< O(\sqrt{n})$, Kasiviswanathan et al.~\cite{kasiviswanathan2010price} prove a higher lower bound\footnote{Their lower bound is under approximate-DP, which also holds under zCDP.} $\tilde{\Omega}(d/n)$, which means that the complexity of covariance estimation is same as $d^2$-dimensional mean estimation in the low-dimensional regime, so one cannot hope to beat the Gaussian mechanism for small $d$. However, in the high-dimensional regime, our result indicates that the covariance problem is strictly easier, due to the correlations of the $d^2$ entries of $\bSigma$. 
Another interesting consequence is that our error bound has utility for $d$ up to $\tilde{O}(n^2)$ (utility is lost when the error is $\tilde{O}(1)$, as returning a zero matrix can already achieve this error). This is the highest $d$ that allows for any utility, since even mean estimation requires $d <\tilde{O}(n^2)$ to have utility under zCDP \cite{KamathLSU19,bun2016concentrated}.  We pictorially show the currently known (worst-case) upper and lower bounds in Figure~\ref{fig:error_bounds}. It remains an interesting open problem to close the gap for $\widetilde{\Omega}(\sqrt{n})<d<\tilde{O}(n^2)$.

\begin{figure}[tbp]
\includegraphics[width=1\textwidth]{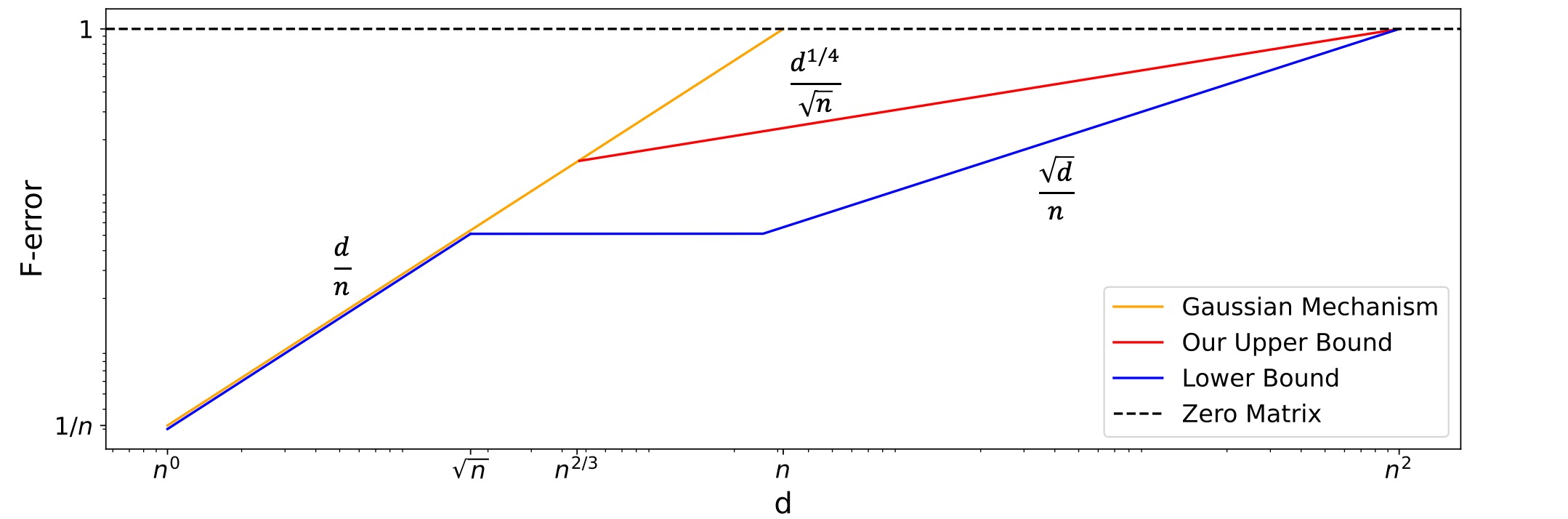}
    \caption{Currently known worst-case error bounds (both axes are in log scale).}
\label{fig:error_bounds}
\end{figure}

Through private communication with Aleksandar Nikolov, it is observed that the \textit{projection mechanism}~\cite{nikolov2013geometry,dwork2015efficient} can also be shown to have error $\tilde{O}(d^{1/4}/\sqrt{n})$ when applied to the covariance problem. In Appendix~\ref{sec:projection}, we make this connection more explicit, while also giving an efficient implementation.
However, the projection mechanism is not trace-sensitive.  

\subsection{A Tail-sensitive Algorithm}
A trace-sensitive bound only makes use of the average $\ell_2$ norm, which cannot capture the full distribution.  Next, we design an algorithm with an error bound that more closely depends on the distribution of the norms.  We characterize this distribution using the \textit{$\tau$-tail} ($\mathbb{I}(\cdot)$ is the indicator function):
\begin{equation}
\label{eq:tail}
\gamma(\X,\tau) = \frac{1}{n}\sum_{i}\|\x_i\|_2^2 \cdot \mathbb{I}(\|\x_i\|_2 > \tau), \tau \in [0,1].
\end{equation}
Note that $\gamma(\X,\tau)$ decreases as $\tau$ increases. In particular, $\gamma(\X, 0)=\mathrm{tr}$, $\gamma(\X,1)=0$.

A common technique to reduce noise, at the expense of some bias, is to clip all the $X_i$'s so that they have norms at most $\tau$, for some threshold $\tau$.  This yields an error of $\mathrm{Noise}(\X, \tau) + \gamma(\X, \tau)$, where $\mathrm{Noise}(\X, \tau)$ denotes the error bound of the mechanism when all the $X_i$'s have norm bounded by $\tau$, and $\gamma(\X,\tau)$ is the (additional) bias caused by clipping.  Opting for the better of the Gaussian mechanism or our trace-sensitive mechanism, we have 
\begin{equation}
\label{eq:error}
\mathrm{Noise}(\X, \tau) = \tilde{O}\left(\min\left(\frac{\tau^2d}{n}, \frac{\tau d^{1/4}\sqrt{\mathrm{tr}}}{\sqrt{n}}+ \frac{\tau^2 \sqrt{d}}{n}\right)\right). 
\end{equation}

The technical challenge is therefore choosing a good $\tau$ in a differentially private manner.  We design a DP mechanism to choose the optimal $\tau$ up to a polylogarithmic multiplicative factor and an exponentially small additive term.  It also adaptively selects the better of Gaussian mechanism or the trace-sensitive mechanism depending on the relationship between $d, n$, and a privatized $\tr$.  More precisely, our adaptive mechanism achieves an error of 
\begin{equation}
\label{eq:ada}
\tilde{O}\left( \min_{\tau}\left(\mathrm{Noise}(\X, \tau) + \gamma(\X, \tau) \right) +2^{-dn}\right).
\end{equation}
Note that this tail-sensitive bound is always no worse (modulo the $2^{-dn}$ term) than $\mathrm{Noise}(\X, 1)$ (i.e., without clipping), and can be much better for certain norm distributions.  In particular, the tail-sensitive bound would work very well on many real datasets with skewed distributions, e.g., most data vectors have small norms with a few having large norms.  For example, suppose $d=n^{3/4}$, and a constant number of data vectors have $\ell_2$ norm $1$ while the others have norm $n^{-1/4}$. Then $\sqrt{\tr} = \Theta(n^{-1/4})$, so $\mathrm{Noise}(\X, 1)$ takes the trace-sensitive bound, which is $\tilde{O}(n^{-9/16})$.  On the other hand, \eqref{eq:ada} is at most $\tilde{O}(n^{-13/16})$ by taking $\tau = n^{-1/4}$.  

\section{Related Work}
\label{sec:related_work}

Mean estimation and covariance estimation are perhaps the most fundamental problems in statistics and machine learning, and how to obtain the best estimates while respecting individual's privacy has attracted a lot of attention in recent years.  Mean estimation under differential privacy is now relatively well understood, with the optimal worst-case error being $\tilde{\Theta}(\sqrt{d}/n)$ \cite{KamathLSU19}, achieved by the standard Gaussian mechanism \cite{dwork2014analyze}.  
In contrast, the covariance problem is more elusive.  As indicated in Figure \ref{fig:error_bounds}, its complexity is probably a piecewise linear (in the log-log scale) function. 

When most data have norms much smaller than the upper bound given \textit{a priori}, the worst-case bounds above are no longer optimal.  In these cases, it is more desirable to have an error bound that is instance-specific. Clipping is a common technique for mean estimation \cite{amin2019bounding,huang2021instance,andrew2021differentially,pichapati2019adaclip,mcmahan2017learning} and it is known that running the Gaussian mechanism after clipping $\X$ with a certain quantile of the norms of the $X_i$'s achieves instance-optimality in a certain sense \cite{amin2019bounding,huang2021instance}.  However, for covariance estimation, we show in Section \ref{sec:clip_cov} that no quantile can be the optimal clipping threshold achieving the bound in \eqref{eq:ada}.  Nevertheless, the bound in \eqref{eq:ada} is only achieving the optimal clipping threshold; we cannot say that is instance-optimal, since $\mathrm{Noise}(\cdot)$ is not even known to be worst-case optimal.

Closely related to covariance estimation are the PCA problem and low-rank approximation.  Instead of finding all eigenvalues and eigenvectors, they only aim at finding the largest one or a few.  For these problems, iterative methods \cite{amin2019differentially,kapralov2013differentially,upadhyay2018price,dwork2014analyze,chaudhuri2013near,singhal2021privately} should perform better than the Gaussian mechanism or our algorithm, both of which try to recover the full covariance matrix.

Many covariance estimation algorithms have been proposed under the statistical setting, where the $X_i$'s are i.i.d.\ samples drawn from a certain distribution, e.g., a multivariate Gaussian \cite{KamathLSU19,bun2019private,biswas2020coinpress,aden2021sample,kamath2022private,liu2022differential,ashtiani2022private,kothari2022private}.  Instead of the Frobenius error, many of them adopt the Mahalanobis error $\|\widetilde{\bSigma} - \bSigma\|_{\bSigma} := \|\bSigma^{-1/2} \widetilde{\bSigma} \bSigma^{-1/2} - \mathbf{I}\|_F$, which can be considered as a normalized version of the former. It is known that $\lambda_d\|\mathbf{A}-\bSigma\|_{\bSigma} \leq \|\mathbf{A}-\bSigma\|_F \leq {\lambda_1}\|\mathbf{A}-\bSigma\|_{\bSigma}$, so when $\bSigma_{\mathbb{D}}$ is well-conditioned, i.e., $\lambda_1/\lambda_d = O(1)$, any Frobenius error directly translates to a Mahalanobis error.  However, for the Mahalanobis error, the more challenging question is how to deal with an ill-conditioned $\bSigma$, for which  \cite{KamathLSU19,biswas2020coinpress} have provided elegant solutions for the case where $\mathbb{D}$ is a multivariate Gaussian.  It would be interesting to see if their methods can be combined with the tail-sensitive techniques in this paper to solve this problem for other distribution families, in particular, heavy-tailed distributions.  For the lower bound, very recently, Kamath et al.~\cite{kamath2022new} proved a similar lower bound for the low-dimensional regime as in \cite{kasiviswanathan2010price} but under the statistical setting.

\section{Preliminaries}
\label{sec:pre}

\subsection{Differential Privacy}
\label{sec:dp}

We say that $\X,\X'\in \mathbb{R}^{d\times n}$ are neighbors if they differ by one column, denoted  $\X\sim \X'$.

\begin{definition}[Differential Privacy (DP) \cite{dwork2006calibrating}]
\label{df:dp}
For $\varepsilon > 0$ and $\delta \ge 0$, a randomized mechanism $\mathcal{M}:\mathbb{R}^{d\times n} \rightarrow \mathcal{Y}$ satisfies $(\varepsilon,\delta)$-DP if for any $\X \sim \X'$ and any $\mathcal{S}\subseteq \mathcal{Y}$, $\Pr[\mathcal{M}(\X)\in \mathcal{S}]\leq e^\varepsilon\cdot \Pr[\mathcal{M}(\X')\in \mathcal{S}]+\delta$.

In particular, we call it \textit{pure-DP} if $\delta=0$; otherwise \textit{approximate-DP}.
\end{definition}

\begin{definition}
[Concentrated Differential Privacy (zCDP) \cite{bun2016concentrated}]
\label{df:cdp}
For $\rho > 0$, a randomized mechanism $M : \mathbb{R}^{d\times n} \rightarrow \mathcal{Y}$ satisfies $\rho$-zCDP if for any $\X \sim \X'$, $D_{\alpha} \left(\mathcal{M}(\X) || \mathcal{M}(\X')\right) \le \rho \cdot \alpha$ for all $\alpha > 1$, where $D_{\alpha}\left(\mathcal{M}(\X) || \mathcal{M}(\X')\right)$ is the $\alpha$-R\'enyi divergence between $\mathcal{M}(\X)$ and $\mathcal{M}(\X')$. 
\end{definition}

The relationship between these DP definitions is as follows. Pure-DP, also written as $\varepsilon$-DP, implies $\frac{\varepsilon^2}{2}$-zCDP, which further implies $\left(\frac{\varepsilon^2}{2} + \varepsilon \sqrt{2 \log \frac{1}{\delta}}, \delta\right)$-DP for any $\delta>0$. 

To preserve $\varepsilon$-DP for a query $Q$, a standard mechanism is to add independent Laplace noises with scale proportional to the (global) $\ell_1$-sensitivity of $Q$ to each dimension.

\begin{lemma}[Laplace Mechanism \cite{dwork2006calibrating}]
\label{lm:lap}
Given $Q: \mathbb{R}^{d\times n} \rightarrow \mathbb{R}^{k}$, let $\mathrm{GS}_Q:=\max_{\X \sim \X'} \| Q(\X) - Q(\X') \|_1$. The mechanism $\mathcal{M}(\X) = Q(\X) + \frac{\mathrm{GS}_Q}{\varepsilon}\cdot \mathbf{Y}$ where $\mathbf{Y}\sim\mathrm{Lap}(1)^k$, preserves $\varepsilon$-DP.
\end{lemma}

The following composition property of $\varepsilon$-DP allows us to design algorithms in a modular fashion.

\begin{lemma}
[Basic Composition]
\label{lm:composition}
If $\mathcal{M}$ is an adaptive composition of mechanisms $\mathcal{M}_1, \mathcal{M}_2, \ldots, \mathcal{M}_t$, where each $\mathcal{M}_i$ satisfies $\varepsilon_i$-DP, then $\mathcal{M}$ satisfies $(\sum_i \varepsilon_i)$-DP.
\end{lemma}

For $\rho$-zCDP, the standard method is the \textit{Gaussian mechanism}:

\begin{lemma}[Gaussian Mechanism \cite{bun2016concentrated}]
\label{lm:gauss}
Given $Q: \mathbb{R}^{d\times n} \rightarrow \mathbb{R}^{k}$, let $\mathrm{GS}_Q:=\max_{\X \sim \X'} \| Q(\X) - Q(\X') \|_2$. The mechanism $\mathcal{M}(\X) = Q(\X) + \frac{\mathrm{GS}_Q}{\sqrt{2\rho}}\cdot\mathbf{Y}$ where $\mathbf{Y}\sim\mathcal{N}\left(0,\mathbf{I}_{k\times k}\right)$, preserves $\rho$-zCDP.
\end{lemma}

It has been shown that the covariance matrix has an $\ell_2$-sensitivity of $\frac{\sqrt{2}}{n}$ \cite{biswas2020coinpress}. Thus, the Gaussian mechanism for covariance, denoted $\mathrm{GaussCov}$, simply adds an independent Gaussian noise with scale ${1 \over \sqrt{\rho} n}$ to each entry of $\bSigma$.  Considering that $\bSigma$ is symmetric, symmetric noises also suffice, which preserve the symmetry of the privatized $\bSigma$.  More precisely, we draw a random noise matrix $\W$ where $w_{j,k}\sim\mathcal{N}(0,1)$ i.i.d. for $1\leq j \leq k \leq d$ and $w_{k,j} = w_{j,k}$, denoted as $\W \sim \mathrm{SGW}(d)$.  Then $\mathrm{GaussCov}$ outputs $\widetilde{\bSigma}_{\mathrm{Gau}} = \bSigma + {1\over \sqrt{\rho} n} \cdot \W$.

A similar composition property exists for $\rho$-zCDP.

\begin{lemma}
[Composition Theorem \cite{bun2016concentrated}]
\label{lm:gau_composition}
If $\mathcal{M}$ is an adaptive composition of algorithms $\mathcal{M}_1, \mathcal{M}_2, \ldots, \mathcal{M}_t$, where each $\mathcal{M}_i$ satisfies $\rho_i$-zCDP, then $\mathcal{M}$ satisfies $(\sum_i \rho_i)$-zCDP.
\end{lemma}

\begin{algorithm}[t]
\caption{$\mathrm{SVT}$}
\label{alg:svt}
\begin{algorithmic}[1]
    \REQUIRE query sequence $(f_1(\X),\dots,f_t(\X))$; threshold $T$; privacy parameter $\varepsilon>0$.
    
    \ENSURE index $k$ 
    
    \STATE $\widetilde{T} \gets T+\mathrm{Lap}(2/\varepsilon)$
    
    \STATE \algorithmicfor \ $k\gets 1,2,\dots,t$
    
    \STATE \hspace{\algorithmicindent} \algorithmicif \ $f_{k}(\X)+\mathrm{Lap}(4/\varepsilon)\geq\widetilde{T}$
    
    \STATE \hspace{\algorithmicindent} \hspace{\algorithmicindent} \textbf{return} $k$\;
    
    \RETURN{$t+1$}
\end{algorithmic}
\end{algorithm}

\subsection{The Sparse Vector Technique}
\label{sec:svt}

The \textit{Sparse Vector Technique} ($\mathrm{SVT}$)~\cite{dwork2009complexity} has as input a sequence of scalar queries, $f_1(\X)\,f_2(\X),\\\dots, f_t(\X)$, where each has sensitivity $1$, and a threshold $T$. It aims to find the first query (if there is) whose answer is approximately above $T$. See Algorithm~\ref{alg:svt} for the details. The $\mathrm{SVT}$ has been shown to satisfy $\varepsilon$-DP with following utility guarantee.

\begin{lemma}[Extension of Theorem 3.24 in \cite{dwork2014algorithmic}]
\label{lm:err_SVT}
With probability at least $1-\beta$, $\mathrm{SVT}$ returns a $k$ such that, for any $i<k$, $f_i(\X) \leq T+\frac{6}{\varepsilon}\log(2t/\beta)$, and if $k\neq t+1$, then $f_{k}(\X)\geq T-\frac{6}{\varepsilon}\log(2t/\beta)$.
\end{lemma}

\subsection{Concentration Inequalities}
\label{sec:concentration_inequality}
\begin{lemma}[\cite{laurent2000adaptive}]
\label{lm:upper_bound_noise}
Given $\mathbf{Y}\sim\mathcal{N}\left(0,\mathbf{I}_{d\times d}\right)$, with probability at least $1-\beta$, 
\[\|\mathbf{Y}\|_2 \leq \eta(d,\beta) := \sqrt{d+2\sqrt{d\log(1/\beta)}+2\log(1/\beta)}.\]
\end{lemma}

\begin{lemma}[\cite{biswas2020coinpress,laurent2000adaptive}]
\label{lm:upper_bound_noise_SGW}
Given $\W\sim \mathrm{SGW}(d)$, with probability at least $1-\beta$, 
\[\|\W\|_2 \leq \upsilon(d,\beta):=2\sqrt{d}+2d^{1/6}\log^{1/3}d + \frac{6(1+(\log d/d)^{1/3})\sqrt{\log d}}{\sqrt{\log(1+(\log d/d)^{1/3}))}}+2\sqrt{2\log(1/\beta)}.\]
Also, with probability at least $1-\beta$, 
\[\|\W\|_F \leq \omega(d,\beta):=\sqrt{d^2+2\sqrt{d\log(2/\beta)}(1+\sqrt{2(d-1)})+6\log(2/\beta)}.\]
\end{lemma}
Ignoring polylogarithmic factors, $\eta(d,\beta)$ and $\upsilon(d,\beta)$ are both in $\tilde{O}(\sqrt{d})$, while $\omega(d,\beta)$ is in $\tilde{O}(d)$.  These concentration inequalities are very useful for error analysis.  For example, the bound on $\|\W\|_F$ immediately implies that $\mathrm{GaussCov}$ has error ${1\over \sqrt{\rho} n} \cdot \omega(d,\beta) = \tilde{O}(d/n)$.

\section{Trace-sensitive Algorithm}
\label{sec:sep}

The state-of-the-art trace-sensitive algorithm~\cite{amin2019differentially} first obtains an estimate of the eigenvalues, and then iteratively finds the eigenvectors by the exponential mechanism (EM), so we denote this algorithm as $\mathrm{EMCov}$.  Under zCDP, it has an error of  $\tilde{O}(d^{3/4}\sqrt{\mathrm{tr}}/\sqrt{n}+\sqrt{d}/n)$.  Below, we present an algorithm that is simpler, faster, and more accurate, improving the trace-dependent term by a $\sqrt{d}$-factor.

The first step of our algorithm $\mathrm{SeparateCov}$ (shown in Algorithm \ref{alg:sep_cov}) is basically the same as $\mathrm{EMCov}$, where we obtain an estimate of the eigenvalues with half of the privacy budget.  \cite{amin2019differentially} uses the Laplace mechanism for pure-DP; for zCDP, we use the Gaussian mechanism, which relies on the $\ell_2$-sensitivity of $\mathbf{\Lambda}$:

\begin{lemma}
\label{lm:eigensens}
For any $\X,\X'\in\mathcal{B}_d^n$, $\X\sim\X'$, 
\begin{equation}
\nonumber
\|\mathbf{\Lambda}-\mathbf{\Lambda}'\|_{F}\leq \frac{\sqrt{2}}{n}.
\end{equation}
\end{lemma}

\begin{proof}
The proof of this lemma requires the following technical result.
\begin{fact}
    \label{claim:vecineq}
    Let $a_k,b_k\geq c_k \geq 0$ for $1\leq k\leq d$. Then
    \begin{equation}
        \nonumber
        \sum_{k=1}^{d}\left(a_k-b_k\right)^2 \leq \left(\sum_{k=1}^d(a_k-c_k)\right)^2 + \left(\sum_{k=1}^d(b_k-c_k)\right)^2.
    \end{equation}
\end{fact}

Let $f_k:\mathcal{B}(r)^n \rightarrow \mathbb{R}_{\geq 0}$ be defined as $f_k(\X) = \lambda_k(\frac{1}{n}\X\X^T)$. Suppose $\X$ and $\X'$ differ in the $i^{\mathrm{th}}$ column. Let $\X_{(-i)}$ denote the matrix obtained from $\X$ by removing the $i^{\mathrm{th}}$ column. Then $f_k(\X) = \lambda_k\left[\frac{1}{n}\X_{(-i)}\X_{(-i)}^T+\frac{1}{n}X_iX_i^T\right]$. Note that $X_iX_i^T$ is a rank-one matrix with eigenvalues $0$ and $\|X_i\|_2^2$, so by Weyl's inequality $f_k(\X_{(-i)}) \leq f_k(\X)$ and $f_k(\X_{(-i)}) \leq f_k(\X')$. We have
\begin{align}
\nonumber
\sum_{k=1}^d\left(f_k(\X)-f_k(\X_{(-i)})\right)=&\sum_{k=1}^d f_k(\X)-\sum_{k=1}^d f_k(\X_{(-i)})
\\
\nonumber
=&\mathrm{tr}\left[\frac{1}{n}\X\X^T\right]-\mathrm{tr}\left[\frac{1}{n}\X_{(-i)}\X_{(-i)}^T\right]
\\
\nonumber
=&\frac{1}{n}\|X_i\|_2^2\leq\frac{1}{n}.
\end{align}
Similarly, $\sum_{k=1}^d\left(f_k(\X')-f_k(\X_{(-i)})\right)\leq \frac{1}{n}$.
Now by Fact \ref{claim:vecineq}, with $a_k=f_k(\X),b_k=f_k(\X')$ and $c_k=f_k(\X_{(-i)})$, we have
\begin{align}
\label{eqn:eigensens_expand}
\nonumber
\sum_{k=1}^d\left(f_k(\X)-f_k(\X')\right)^2 &\leq \left(\sum_{k=1}^d\left(f_k(\X)-f_k(\X_{(-i)})\right)\right)^2 + \left(\sum_{k=1}^d\left(f_k(\X')-f_k(\X_{(-i)})\right)\right)^2 
\\
&\leq \left(\frac{1}{n}\right)^2 + \left(\frac{1}{n}\right)^2 = \frac{2}{n^2}.
\end{align}
Taking the square root on both sides then gives the target inequality.
\end{proof}

For the eigenvectors, we use $\mathrm{GaussCov}$ to obtain a privatized $\widetilde{\mathbf{\Sigma}}_{\mathrm{Gau}}$ with the other half of the privacy budget, and perform an eigendecomposition.  Finally, we assemble the eigenvalues of eigenvectors to obtain a privatized $\bSigma$.  It should be clear that, after computing $\bSigma$, $\mathrm{SeparateCov}$ just needs two eigendecompositions and one full matrix multiplication, plus some $O(d^2)$-time operations.  On the other hand, $\mathrm{EMCov}$ performs $O(d)$ eigendecompositions and matrix multiplications, plus a nontrivial sampling procedure for the EM.

\begin{algorithm}
\caption{$\mathrm{SeparateCov}$}
\label{alg:sep_cov}
\begin{algorithmic}[1]
    \REQUIRE data $\X\in \mathcal{B}_d^n$; privacy parameter $\rho>0$.

    \STATE $\mathbf{\Lambda} \gets$ the eigenvalues of $\bSigma={1\over n} \X \X^T$
    
    \STATE $\widetilde{\mathbf{\Lambda}}_{\mathrm{Sep}} \gets \Lam + \frac{\sqrt{2}}{\sqrt{\rho} n} \cdot \mathbf{Y}$, where $\mathbf{Y}\sim\mathcal{N}(0, \mathbf{I}_{d\times d})$
    
    \STATE $\widetilde{\mathbf{\Sigma}}_{\mathrm{Gau}}\gets \mathrm{GaussCov}(\X,\frac{\rho}{2})$
    
    \STATE $\widetilde{\mathbf{P}}_{\mathrm{Sep}}\gets \mathbf{P}\left[\widetilde{\mathbf{\Sigma}}_{\mathrm{Gau}}\right]$
    
    \STATE $\widetilde{\mathbf{\Sigma}}_{\mathrm{Sep}} \gets\widetilde{\mathbf{P}}_{\mathrm{Sep}} \widetilde{\mathbf{\Lambda}}_{\mathrm{Sep}} \widetilde{\mathbf{P}}_{\mathrm{Sep}}^T$
    
    \RETURN{$\widetilde{\mathbf{\Sigma}}_{\mathrm{Sep}}$}
\end{algorithmic}
\end{algorithm}

That $\mathrm{SeparateCov}$ satisfies $\rho$-zCDP easily follows from the privacy of the Gaussian mechanism and the composition property.  The utility is given by the following theorem:

\begin{theorem}
\label{th:err_sep_cov}
Given any $\rho>0$, for any $\X\in \mathcal{B}_d^n$, and any $\beta>0$, with probability at least $1-\beta$, $\mathrm{SeparateCov}$ returns a $\widetilde{\mathbf{\Sigma}}_{\mathrm{Sep}}$ such that $\|\widetilde{\mathbf{\Sigma}}_{\mathrm{Sep}}-\mathbf{\Sigma}\|_F \le \frac{2^{1.25}\sqrt{\mathrm{tr}}}{\rho^{1/4}\sqrt{n}}\cdot \sqrt{\upsilon\left(d,{\beta\over 2}\right)}+\frac{\sqrt{2}}{\sqrt{\rho}n}\cdot \eta\left(d,{\beta\over 2}\right) = \tilde{O}\left(\frac{d^{1/4}\sqrt{\mathrm{tr}}}{\sqrt{n}}+\frac{\sqrt{d}}{n}\right)$.
\end{theorem}

\begin{proof}
The error can be decomposed into two parts:
\begin{align*}
\|\widetilde{\mathbf{\Sigma}}_{\mathrm{Sep}}-\mathbf{\Sigma}\|_F =& \|\widetilde{\mathbf{P}}_{\mathrm{Sep}} \widetilde{\mathbf{\Lambda}}_{\mathrm{Sep}} \widetilde{\mathbf{P}}_{\mathrm{Sep}}^T - \mathbf{\Sigma}\|_F
\\
=& \|\widetilde{\mathbf{P}}_{\mathrm{Sep}} (\widetilde{\mathbf{\Lambda}}_{\mathrm{Sep}}-\Lam) \widetilde{\mathbf{P}}_{\mathrm{Sep}}^T + \widetilde{\mathbf{P}}_{\mathrm{Sep}} \Lam \widetilde{\mathbf{P}}_{\mathrm{Sep}}^T - \mathbf{\Sigma}\|_F
\\
\leq &  \|\widetilde{\mathbf{P}}_{\mathrm{Sep}} (\widetilde{\mathbf{\Lambda}}_{\mathrm{Sep}}-\Lam) \widetilde{\mathbf{P}}_{\mathrm{Sep}}^T \|_F + \|\widetilde{\mathbf{P}}_{\mathrm{Sep}} \Lam \widetilde{\mathbf{P}}_{\mathrm{Sep}}^T - \mathbf{\Sigma}\|_F.
\end{align*}

For the first term, by Lemma~\ref{lm:upper_bound_noise}, we have with probability at least $1-\frac{\beta}{2}$,
\[\|\widetilde{\mathbf{P}}_{\mathrm{Sep}} (\widetilde{\mathbf{\Lambda}}_{\mathrm{Sep}}-\Lam) \widetilde{\mathbf{P}}_{\mathrm{Sep}}^T \|_F =\|\widetilde{\mathbf{\Lambda}}_{\mathrm{Sep}}-\Lam\|_F =\frac{\sqrt{2}}{\sqrt{\rho}n} \cdot\|\mathbf{Y}\|_2  = \frac{\sqrt{2}}{\sqrt{\rho}n}\cdot \eta(d,\beta).\]

For the second term, 
\begin{align*}
\|\widetilde{\mathbf{P}}_{\mathrm{Sep}} \Lam \widetilde{\mathbf{P}}_{\mathrm{Sep}}^T - \mathbf{\Sigma}\|_F^2 = & \|\widetilde{\mathbf{P}}_{\mathrm{Sep}} \Lam \widetilde{\mathbf{P}}_{\mathrm{Sep}}^T\|_F^2 +\|\mathbf{\Sigma}\|_F^2 - 2\cdot \mathrm{tr}\left(\mathbf{\Sigma}\widetilde{\mathbf{P}}_{\mathrm{Sep}} \Lam \widetilde{\mathbf{P}}_{\mathrm{Sep}}^T\right)
\\
= & 2 \sum_j \lambda_j^2 - 2\cdot\mathrm{tr}\left(\Lam \widetilde{\mathbf{P}}_{\mathrm{Sep}}^T\mathbf{\Sigma}\widetilde{\mathbf{P}}_{\mathrm{Sep}}\right)
\\
= & 2 \sum_j \lambda_j^2 - 2\sum_{j} \lambda_j\left(\widetilde{P}_{\mathrm{Sep},j}^T\mathbf{\Sigma}\widetilde{P}_{\mathrm{Sep},j}\right)
\\
= & 2\sum_{j}\lambda_j \left(\lambda_j-\widetilde{P}_{\mathrm{Sep},j}^T\mathbf{\Sigma}\widetilde{P}_{\mathrm{Sep},j}\right).
\end{align*}

Now the only work is to bound the terms $\left(\lambda_j-\widetilde{P}_{\mathrm{Sep},j}^T\mathbf{\Sigma}\widetilde{P}_{\mathrm{Sep},j}\right)$ for all $j$. First, recall $\widetilde{\mathbf{\Sigma}}_{\mathrm{Gau}} = \mathbf{\Sigma}+\frac{1}{\sqrt{\rho}n}\cdot \mathbf{W}$, where $\mathbf{W}\sim \mathrm{SGW}(d)$. By Lemma~\ref{lm:upper_bound_noise_SGW}, with probability at least $1-\frac{\beta}{2}$, for all unit vector $u$,
\begin{equation}
\label{eq:lm:err_sep_cov_1}
\left|u^T\widetilde{\mathbf{\Sigma}}_{\mathrm{Gau}}u - u^T\mathbf{\Sigma}u\right| = \frac{1}{\sqrt{\rho}n}\cdot \left|u^T \mathbf{W} u\right| \leq \frac{\sqrt{2}}{\sqrt{\rho}n}\cdot\upsilon(d,\beta/2).
\end{equation}

For any $j$, let $\mathcal{S}_j = \mathrm{span}\{P_{j+1},\dots,P_d\}$, and $\mathcal{S}_j' = \mathrm{span}\{\widetilde{P}_{\mathrm{Sep},j+1},\dots,\widetilde{P}_{\mathrm{Sep},d}\}$. Let
\[\lambda'_j = \max_{u:\|u\|_2=1, u\in\mathcal{S}_j'}u^T\mathbf{\Sigma} u\]
and
\[P'_j=\argmax_{u:\|u\|_2=1, u\in\mathcal{S}_j'}u^T\mathbf{\Sigma} u.\]
By definition,
\begin{equation}
\label{eq:lm:err_sep_cov_2}
\lambda_j = \max_{u:\|u\|_2=1, u\in\mathcal{S}_j}u^T\mathbf{\Sigma} u\leq  \max_{u:\|u\|_2=1, u\in\mathcal{S}_j'}u^T\mathbf{\Sigma} u = \lambda'_j.
\end{equation}
Finally, we have
\begin{align*}
\widetilde{P}_{\mathrm{Sep},j}^T\mathbf{\Sigma}\widetilde{P}_{\mathrm{Sep},j}\geq & \widetilde{P}_{\mathrm{Sep},j}^T\widetilde{\mathbf{\Sigma}}_{\mathrm{Gau}}\widetilde{P}_{\mathrm{Sep},j} - \frac{\sqrt{2}}{\sqrt{\rho}n}\cdot \upsilon(d,\beta)
\\
\geq & P_j^{'T}\widetilde{\mathbf{\Sigma}}_{\mathrm{Gau}}P'_j- \frac{\sqrt{2}}{\sqrt{\rho}n}\cdot \upsilon(d,\beta)
\\
\geq & P_j^{'T}\mathbf{\Sigma}P'_j-\frac{2\sqrt{2}}{\sqrt{\rho}n}\cdot \upsilon(d,\beta)
\\
\geq & \lambda_j - \frac{2\sqrt{2}}{\sqrt{\rho}n}\cdot \upsilon(d,\beta).
\end{align*}
The first and third inequalities are from (\ref{eq:lm:err_sep_cov_1}). The second inequality is by the definition of $\widetilde{P}_{\mathrm{Sep},j}$. The last inequality is by (\ref{eq:lm:err_sep_cov_2}).
\end{proof}

\paragraph{Remark} While $\mathrm{SeparateCov}$ strictly improves over $\mathrm{EMCov}$, it does not dominate $\mathrm{GaussCov}$: When $\mathrm{tr}< \tilde{O}(d^{3/2}/n)$, $\mathrm{SeparateCov}$ is better; otherwise, $\mathrm{GaussCov}$ is better.   $\mathrm{EMCov}$ is better than $\mathrm{GaussCov}$ for a smaller trace range: $\tr < \tilde{O}(\sqrt{d}/n)$. 

\medskip
Theorem~\ref{th:err_sep_cov} implies our worst-case bound by taking $\mathrm{tr} = 1$:

\begin{theorem}
\label{th:err_sep_cov_worst_case}
Given any $\rho>0$, for any $\X\in \mathcal{B}_d^n$, and any $\beta>0$, with probability at least $1-\beta$, $\mathrm{SeparateCov}$ returns a $\widetilde{\mathbf{\Sigma}}_{\mathrm{Sep}}$ such that $\|\widetilde{\mathbf{\Sigma}}_{\mathrm{Sep}}-\mathbf{\Sigma}\|_F = \tilde{O}\left(\frac{d^{1/4}}{\sqrt{n}}+\frac{\sqrt{d}}{n}\right)$.
\end{theorem}

\section{Tail-sensitive Algorithm}
\label{sec:cov_sparse_data}

\subsection{Clipped Covariance}
\label{sec:clip_cov}
Clipping is a common technique to reduce the sensitivity of functions at the expense of some bias.  Given $\tau\ge 0$ and a vector $\x \in \mathbb{R}^d$, let $\mathrm{Clip}(\x,\tau) = \min\left(1, \frac{\tau}{\|\x\|_2}\right)\cdot \x$. Similarly, for any $\X\in\mathbb{R}^{d\times n}$, $\mathrm{Clip}(\X,\tau)$ denotes the matrix whose columns have been clipped to have norm at most $\tau$.  Clipping can be applied to both $\mathrm{GaussCov}$ and $\mathrm{SeparateCov}$ with a given $\tau$: just run the mechanism on $\frac{1}{\tau}\cdot\mathrm{Clip}(\X,\tau)$ and scale the result back by $\tau^2$. We denote the clipped versions of the two mechanisms as  $\mathrm{ClipGaussCov}$ and $\mathrm{ClipSeparateCov}$, respectively.

The following lemma bounds the bias caused by clipping in terms of the $\tau$-tail as defined in \eqref{eq:tail}.

\begin{lemma}
\label{lm:bounded_bias}
$\left\| \bSigma(\X)-\bSigma(\mathrm{Clip}(\X,\tau))\right\|_F \leq {1\over n} \sum_i \left(\|\x_i\|_2^2-\tau^2\right)\cdot \mathbb{I}\left(\|\x_i\|_2\geq \tau\right)\leq \gamma(\X,\tau)$.
\end{lemma}

\begin{proof}
The proof of this lemma can be derived by the following technical result.

\begin{claim}
For any $\x\in \mathbf{R}^d$, and $\tau\leq \|\x\|_2$, let $\widecheck{\x} = \mathrm{Clip}(\x,\tau)$, then,
\[\left\|\x\x^T-\widecheck{\x}\widecheck{\x}^{T}\right\|_F = \|\x\|_2^2-\tau^2.\]
\end{claim}

By definition, $\widecheck{\x} = \frac{\tau}{\|\x\|_2}\x$. Then
\begin{align*}
\left\|\x\x^T-\widecheck{\x}\widecheck{\x}^{T}\right\|_F =& \left\|\x\x^T-(\frac{\tau}{\|\x\|_2})^2\x\x^{T}\right\|_F
\\
 =& \frac{\|\x\|_2^2-\tau^2}{\|\x\|_2^2} \|\x\x^T\|_F
\\
 =& \|\x\|_2^2-\tau^2.
\end{align*}
\end{proof}

Thus, running the better of $\mathrm{ClipGaussCov}$ and $\mathrm{ClipSeparateCov}$ yields a total error of $\mathrm{Noise}(\X,\tau,\rho,\beta) + \gamma(\X, \tau)$, where
\begin{equation}
\label{eq:error_detail}
\mathrm{Noise}(\X,\tau,\rho,\beta) = \min\left(\frac{\tau^2}{\sqrt{\rho} n}\cdot\omega(d,\beta),\frac{2^{1.25}\tau\sqrt{\mathrm{tr}}}{\rho^{1/4}\sqrt{n}}\cdot\sqrt{\upsilon\left(d,\frac{\beta}{2}\right)}+ \frac{\sqrt{2}\tau^2}{\sqrt{\rho}n}\cdot\eta\left(d,\frac{\beta}{2}\right)\right),
\end{equation}
which is the exact version of (\ref{eq:error}).  Note that the trace-sensitive term is only scaled by $\tau$, which follows from the proof of Theorem \ref{th:err_sep_cov} when all $X_i$ live in $\tau \cdot \mathcal{B}_d$. 

Ideally, we would like to find the optimal noise-bias trade-off, i.e., achieving an error of $\min_\tau(\mathrm{Noise}(\X,\tau) + \gamma(\X, \tau))$.  Two issues need to be addressed towards this goal:  The first, minor, issue is that $\tr$ is sensitive, so we cannot use it directly to decide whether to use $\mathrm{ClipGaussCov}$ or $\mathrm{ClipSeparateCov}$.  This can be addressed by using a privatized upper bound of $\tr$.  The more challenging problem is how to find the optimal $\tau$ in a DP fashion.  This problem has been well studied for the clipped mean estimator \cite{huang2021instance,amin2019bounding}, where it can be shown that setting $\tau$ to be the  $\tilde{O}(\sqrt{d})$-th largest $\|X_i\|_2$ results in the optimal noise-bias trade-off
\cite{huang2021instance}.  Then the problem boils down to finding a privatized quantile, for which multiple solutions exist \cite{huang2021instance,dong2021universal,nissim2007smooth,asi2020instance,smith2011privacy}.  For the clipped mean estimator, using such a quantile of the norms results in the optimal trade-off because $\mathrm{Noise}(\X,\tau)$ takes the simple form $\tilde{O}(\tau \sqrt{d}/n)$.  In fact, if we only had $\mathrm{ClipGaussCov}$, setting $\tau$ to be the $\tilde{O}(d)$-th largest $\|X_i\|_2$ would also yield an optimal trade-off, as $\mathrm{ClipGaussCov}$ is really just clipped mean in $d^2$ dimensions.  However, due to the trace-sensitive noise term, it is no longer the case.  Here, we give an example showing that no quantile, whose rank may arbitrarily depend on $d,n,\tr$, can achieve an optimal trade-off even ignoring polylogarithmic factors.  It thus calls for a new threshold-finding mechanism, which we describe next.

\begin{example}
We construct two datasets with the same $d$, $n$, and $\mathrm{tr}$ whose optimal quantile-based thresholds are asymptotically different. Let $d = \Theta(n^{9/16})$, $\mathrm{tr} = \frac{d}{n} = \Theta(n^{-7/16})$. With such setting, $\mathrm{SeparateCov}$ performs better than $\mathrm{GaussCov}$. For the first dataset, we have $n^{9/16}$ number of vectors with $\ell_2$ norm equal to one while all others have norm zero. In this case, we have $1$ and $0$ as the only candidate truncation thresholds. Using $\tau = 1$ gives an error of $\Theta(n^{-37/64})$ and while $\tau = 0$ corresponds to an error of $\Theta(n^{-7/16})$. Thus, $\tau = 1$ should be used, which corresponds a quantile between the $0$th and $n^{9/16}$th. For the second dataset, we have $O(1)$ number of vectors with $\ell_2$ norm square equal to one, $n^{19/32}$ number of vectors with $\ell_2$ norm square equal to $n^{-13/32}$ and all others have $\ell_2$ norm square equal to $n^{-7/16}$. Here, we have three candidates: $1$, $n^{-13/64}$ and $n^{-7/32}$. $\tau=1$ corresponds to the error $\Theta(n^{-37/64})$. $\tau = n^{-13/64}$ corresponds to the error $\Theta(n^{-25/32})$. $\tau = n^{-7/32}$ corresponds to the error $\Theta(n^{-51/64})$. Thus, we should choose $\tau = n^{-7/32}$, which corresponds to a quantile between the  $(\tilde{O}(1)+n^{19/32})$th and $n$th.
\end{example}

\subsection{Adaptive Covariance: Finding the Optimal Clipping Threshold}

Our basic idea is to try successively smaller values $\tau = 1,\frac{1}{2},\frac{1}{4},\dots$. As we reduce $\tau$, the noise decreases while the bias increases.  We should stop when they are approximately balanced, which would yield a near-optimal $\tau$. 

To do so in a DP manner, we need to quantify the noise and bias.  Consider the bias first.  Given a $\tau$, we divide the interval $(\tau,1]$ into sub-intervals $(\tau,2\tau],(2\tau,4\tau],\dots,(\frac{1}{2},1]$. For any $\x\in\X$ such that $\|\x\|_2\in (2^{s},2^{s+1}]$, let $\widecheck{\x} = \mathrm{Clip}(\x,\tau)$ and then by Lemma~\ref{lm:bounded_bias},
\begin{equation}
\label{eq:bias_1}
\|\x\x^T-\widecheck{\x}\widecheck{\x}^{T}\|_F \leq 2^{2s+2} - \tau^2.
\end{equation}
That is, clipping $\x$ can at most lead to $\frac{1}{n}\cdot (2^{2s+2}-\tau^2)$ bias. Besides, since $\|\x\|_2\in (2^{s},2^{s+1}]$, we have
\begin{equation}
\label{eq:bias_2}
2^{2s+2} - \tau^2\leq 2^{2s+2} \leq 2\cdot \|\x\|_2^2.
\end{equation}
Then, given ${\X}$, for any $s\in \mathbb{Z}$, we define $\mathrm{Count}_s({\X}) := \left|\left\{{\x}_i:\|{\x}_i\|_2\in (2^{s},2^{s+1}]\right\}\right|$.
It is easy to see for any $\X\sim\X'$, $\mathrm{Count}_s$ differs by at most $1$, so does the sum of any subset of $\mathrm{Count}_s$'s. We can define an upper bound on the bias: $\widehat{\mathrm{Bias}}({\X},\tau):=\frac{1}{n}\cdot\sum_{s=\log_2(\tau)}^{s<0} \mathrm{Count}_s \cdot (2^{2s+2}-\tau^2)$. Let $\widecheck{\X} = \mathrm{Clip}({\X},\tau)$. By (\ref{eq:bias_1}) and (\ref{eq:bias_2}), we have
\begin{equation}
\label{eq:bias_3}
{1\over n} \|{\X}{\X}^T-\widecheck{\X}\widecheck{\X}^T\|\leq \widehat{\mathrm{Bias}}({\X},\tau) \leq 2\cdot \gamma({\X},\tau).
\end{equation}
By the property of $\mathrm{Count}_s$'s, given any $\tau$, the sensitivity of $\widehat{\mathrm{Bias}}(\cdot,\tau)$ is bounded by $\frac{1}{n}$.

Now we turn to the noise. Recall that $\mathrm{Noise}({\X},\tau,\rho,\beta)$ is the smaller of two parts. The first part $\mathrm{GaussNoise}(\tau,\rho,\beta) := \tau^2\cdot\frac{1}{\sqrt{\rho} n}\cdot\omega(d,\beta)$ is independent of $\X$, so can be used directly.  The second part depends on $\tr$, is thus sensitive.  Since its sensitivity is $\frac{1}{n}$, we can easily privatize it by adding a Gaussian noise of scale $\Theta\left({1\over \sqrt{\rho}n}\right)$.  For technical reasons, we need to use an upper bound, so we add $\Theta\left({\log(1/\beta) \over \sqrt{\rho} n}\right)$ to it so as to obtain a privatized $\widehat{\tr} \ge \tr$.  Then we set
\begin{equation*}
\mathrm{SeparateNoise}(\widehat{\mathrm{tr}},\tau,\rho,\beta) :=\tau\cdot \frac{2^{1.25}\sqrt{\widehat{\mathrm{tr}}}}{\rho^{1/4}\sqrt{n}}\cdot\sqrt{\upsilon\left(d,\frac{\beta}{2}\right)}+\tau^2\cdot \frac{\sqrt{2}}{\sqrt{\rho}n}\cdot\eta\left(d,\frac{\beta}{2}\right),
\end{equation*}
and use $\widehat{\mathrm{Noise}}(\widehat{\mathrm{tr}},\tau,\rho,\beta) := \min\left(\mathrm{GaussNoise}(\tau,\rho,\beta),\mathrm{SeparateNoise}(\widehat{\mathrm{tr}},\tau,\rho,\beta)\right)$ as a DP upper bound of $\mathrm{Noise}({\X},\tau,\rho,\beta)$.  Note that given $\widehat{\mathrm{tr}}$, $\widehat{\mathrm{Noise}}(\widehat{\mathrm{tr}},\tau,\rho,\beta)$ is independent of $\X$.

\begin{algorithm}[t]
\caption{$\mathrm{AdaptiveCov}$}
\label{alg:Ada_cov}
\begin{algorithmic}[1]
    \REQUIRE data $\X\in \mathcal{B}_d^n$; privacy parameter $\rho>0$; high probablity parameter $\beta$.
    \STATE $\tilde{r} \gets \mathrm{PrivRadius}(\X,\frac{\sqrt{\rho}}{2},\frac{\beta}{8}, 2^{-2dn})$
    
    \STATE $\widetilde{\X}\gets\mathrm{Clip}(\X,\tilde{r})$
    
    \STATE $\widetilde{\mathrm{tr}}\gets\frac{1}{n}\sum_i\|\widetilde{\x}_i\|_2^2$
    
    \STATE $\widehat{\mathrm{tr}} \gets\min\left( \widetilde{\mathrm{tr}}+\frac{2\tilde{r}^2}{\sqrt{\rho} n}\cdot \mathcal{N}(0,1)+\frac{2\sqrt{2}\tilde{r}^2}{\sqrt{\rho}n}\cdot\sqrt{\log(8/\beta)},\ \tilde{r}^2\right)$
    
    \STATE $\tilde{t} \gets \log_2(\tilde{r})+1-\mathrm{SVT}\left(\left\{\mathrm{Diff}\left(\widetilde{\X},\widehat{\mathrm{tr}},\tilde{r},\frac{\rho}{2},\frac{\beta}{2}\right),\mathrm{Diff}\left((\widetilde{\X},\widehat{\mathrm{tr}},\frac{\tilde{r}}{2},\frac{\rho}{2},\frac{\beta}{2}\right),\dots,\mathrm{Diff}\left((\widetilde{\X},\widehat{\mathrm{tr}},2^{-dn},\frac{\rho}{2},\frac{\beta}{2}\right)\right\},0,\frac{\sqrt{\rho}}{\sqrt{2}}\right)$
    
    \STATE $\tilde{\tau} \gets \min\left(2^{\tilde{t}+1},\tilde{r}\right)$
    
    \STATE \algorithmicif \  $\mathrm{SeparateNoise}(\widehat{\mathrm{tr}},\tilde{\tau},\frac{\rho}{2},\frac{\beta}{2})\geq \mathrm{GaussNoise}(\tilde{\tau},\frac{\rho}{2},\frac{\beta}{2})$
    
    \STATE \hspace{\algorithmicindent} $\widetilde{\mathbf{\Sigma}}_{\mathrm{Ada}} \gets \mathrm{ClipGaussCov}(\widetilde{\X},\frac{\rho}{2},\tilde{\tau})$
    
    \STATE \algorithmicelse
    
    \STATE \hspace{\algorithmicindent} $\widetilde{\mathbf{\Sigma}}_{\mathrm{Ada}} \gets \mathrm{ClipSeparateCov}(\widetilde{\X},\frac{\rho}{2},\tilde{\tau})$
    
    \RETURN{$\widetilde{\mathbf{\Sigma}}_{\mathrm{Ada}}$}
\end{algorithmic}
\end{algorithm}

Finally, we run SVT on the following sequence of sensitivity-1 queries with $T=0$:
\begin{equation*}
\label{eq:diff}
\mathrm{Diff}({\X},\widehat{\mathrm{tr}},\tau,\rho,\beta):=n\cdot\left(\widehat{\mathrm{Bias}}({\X},\tau) - \widehat{\mathrm{Noise}}(\widehat{\mathrm{tr}},\tau,\rho,\beta)\right), \tau=1,\frac{1}{2},\dots,2^{-dn}.
\end{equation*}
The SVT would return a $\tau$ that balances the bias and noise.  After finding such a $\tau$, we choose to run either $\mathrm{GaussCov}$ or $\mathrm{SeparateCov}$ by comparing $\mathrm{GaussNoise}(\tau,\rho,\beta)$ and $\mathrm{SeparateNoise}(\widehat{\mathrm{tr}},\tau,\rho,\beta)$.  As the sequence consists of $dn$ queries, SVT has an error of $O(\log(dn))$, which, as we will show, affects the optimality by a logarithmic factor.  Meanwhile, the smallest $\tau$ we search over will induce an additive $2^{-dn}$ error.  

The algorithm above can almost give us the desired error bound in \eqref{eq:ada}, except that one thing may go wrong: The SVT introduces an error that is a logarithmic factor larger than the optimum, but at least $\widetilde{\Omega}(1/n)$.  This would be fine as long as there is one $X_i$ with $\|X_i\|_2\ge \widetilde{\Omega}(1)$, so that the optimum error is $\widetilde{\Omega}(1/n)$.  However, when all the $X_i$'s have very small norms, say $1/n^2$, the $\widetilde{\Omega}(1/n)$ error from SVT would not preserve optimality.  To address this issue, we first find the radius $\mathrm{rad}(\X) = \max_i \|X_i\|_2$, and use it to clip $\X$.  The following lemma shows that, under DP, it is possible to find a 2-approximation of $\mathrm{rad}(\X)$ plus an additive $b$ so that only $O(\log\log(1/b))$ vectors are clipped.  This allows us to set $b=2^{-dn}$ while only incurring an $O(\log dn)$ error.  Nicely, they match the additive and multiplicative errors that already exist from the SVT, so there is no asymptotic degradation in the optimality.

\begin{lemma}[\cite{dong2021universal}]
\label{lm:err_radius}
For any $\varepsilon>0$, $\beta>0$ and $b>0$, given $\X\in \mathcal{B}_d^{n}$, with probability at least $1-\beta$, $\mathtt{PrivRadius}$ returns a $\tilde{r}=\mathtt{PrivRadius}(\X,\varepsilon,\beta,b)$ such that $\tilde{r}\leq 2\cdot \mathrm{rad}(\X)+b$ and $\left|\left\{\|\x_i\|_2> \tilde{r}\right\}\right| = O\left(\frac{1}{\varepsilon}\log\frac{\log(\mathrm{rad}(\X)/b)}{\beta}\right).$
\end{lemma}

The complete algorithm is given in Algorithm~\ref{alg:Ada_cov}.  Its privacy follows from the privacy of $\mathtt{PrivRadius}$, SVT,  $\mathrm{GaussCov}$, $\mathrm{SeparateCov}$, and the composition theorem of zCDP; its utility is analyzed in the following theorem:

\begin{theorem}
\label{th:err_ada_cov}
Given any $\rho>0$ and $\beta>0$, for any $\X\in\mathcal{B}_d^n$, with probability at least $1-\beta$, $\mathrm{AdaptiveCov}$ returns a $\widetilde{\mathbf{\Sigma}}_{\mathrm{Ada}}$ such that
\begin{align*}
\left\|\widetilde{\mathbf{\Sigma}}_{\mathrm{Ada}}-\mathbf{\Sigma}\right\|_F=&O\left(\min_{\tau}\left(\mathrm{Noise}\left(\X,\tau,\frac{\rho}{2},\frac{\beta}{2}\right)\cdot \frac{\log(1/\beta)^{1/4}}{\rho^{1/4}}\cdot\log(n)+\gamma(\X,\tau)\cdot \frac{\log(dn/\beta)}{\sqrt{\rho}}\right)+2^{-dn}\right)
\\
=&\tilde{O}\left(\min_{\tau}\left(\mathrm{Noise}\left(\X,\tau\right)+\gamma(\X,\tau)\right)+2^{-dn}\right).
\end{align*}
\end{theorem}

\begin{proof}
We start by analyzing the error from using the private estimate $\tilde{r}$ of $\mathrm{rad}(\X)$. By Lemma~\ref{lm:err_radius}, we have, with probability at least $1-\frac{\beta}{8}$,
\begin{equation}
\label{eq:th:err_ada_cov_1}
\tilde{r}\leq 2\cdot \mathrm{rad}(\X)+2^{-2dn},
\end{equation}
and
\begin{equation*}
\left|\left\{\|\x_i\|_2> \tilde{r}\right\}\right| = O\left(\frac{1}{\sqrt{\rho}}\log\frac{dn}{\beta}\right).
\end{equation*}
This gives
\begin{equation}
\label{eq:th:err_ada_cov_2}
\left\|\mathbf{\Sigma}-\mathbf{\Sigma}[\widetilde{\X}] \right\|_F = O\left(\frac{\mathrm{rad}(\X)^2}{\sqrt{\rho }n}\log\frac{dn}{\beta}\right).
\end{equation}

Now, we analyze the error  $\left\|\widetilde{\mathbf{\Sigma}}_{\mathrm{Ada}}-\mathbf{\Sigma}[\widetilde{\X}] \right\|_F$. First, by tail bound of Gaussian, with probability at least $1-\frac{\beta}{8}$,
\begin{equation}
\label{eq:th:err_ada_cov_3}
\widetilde{\mathrm{tr}} \leq \widehat{\mathrm{tr}}\leq \widetilde{\mathrm{tr}} + \frac{4\sqrt{2}\cdot\tilde{r}^2}{\sqrt{\rho}n}\cdot\sqrt{\log(8/\beta)}.
\end{equation}

Second, as the discussion in the definitions of $\mathrm{GaussNoise}(\cdot)$, $\mathrm{SeparateNoise}(\cdot)$ and $\widehat{\mathrm{Noise}}(\cdot)$, with probability at least $1-\frac{\beta}{2}$, we have
\begin{equation}
\label{eq:th:err_ada_cov_4}
\left\|\widetilde{\mathbf{\Sigma}}_{\mathrm{Ada}} - \mathbf{\Sigma}[\widecheck{\X}]\right\|_F = O\left(\widehat{\mathrm{Noise}}\left(\widetilde{\X},\widehat{\mathrm{tr}},\tilde{\tau},\frac{\rho}{2},\frac{\beta}{2}\right)\right).
\end{equation}
By (\ref{eq:bias_3}),
\begin{equation}
\label{eq:th:err_ada_cov_5}
\left\|\mathbf{\Sigma}[\widetilde{\X}]-\mathbf{\Sigma}[\widecheck{\X}]\right\|_F \leq \widehat{\mathrm{Bias}}\left(\widetilde{\X},\tilde{\tau}\right).
\end{equation}
Combining (\ref{eq:th:err_ada_cov_4}) and (\ref{eq:th:err_ada_cov_5}), we have
\begin{equation}
\label{eq:th:err_ada_cov_6}
\left\|\widetilde{\mathbf{\Sigma}}_{\mathrm{Ada}}-\mathbf{\Sigma}[\widetilde{\X}]\right\|_F = O\left(\widehat{\mathrm{Bias}}\left(\X,\tilde{\tau}\right) +\widehat{\mathrm{Noise}}\left(\X,\widehat{\mathrm{tr}},\tilde{\tau},\frac{\rho}{2},\frac{\beta}{2}\right)\right).
\end{equation}

Then, let's analyze the error from $\mathrm{SVT}$. There are two sub-cases: $\tilde{t} = -dn-1$ and $\tilde{t}\geq -dn$. For the first case, $\tilde{\tau} = 2^{-dn}$ and no query is higher than the threshold. By Lemma~\ref{lm:err_SVT}, with probability at least $1-\frac{\beta}{4}$, for any $\tau = \tilde{r},\frac{\tilde{r}}{2},\dots,2^{-dn}$,
\[\mathrm{Diff}\left(\widetilde{\X},\widehat{\mathrm{tr}},\tau,\frac{\rho}{2},\frac{\beta}{2}\right) = O\left(\frac{1}{\sqrt{\rho}}\log(dn/\beta)\right),\]
Recall that $\mathrm{Diff}\left(\widetilde{\X},\widehat{\mathrm{tr}},\tau,\frac{\rho}{2},\frac{\beta}{2}\right)$ increases with the decrease of $\tau$, we have, for all $\tau = \tilde{r},\frac{\tilde{r}}{2},\dots,2^{-dn}$,
\begin{equation}
\label{eq:th:err_ada_cov_7}
\left| \mathrm{Diff}\left(\widetilde{\X},\widehat{\mathrm{tr}},\tilde{\tau},\frac{\rho}{2},\frac{\beta}{2}\right) \right|- \left|\mathrm{Diff}\left(\widetilde{\X},\widehat{\mathrm{tr}},\tau,\frac{\rho}{2},\frac{\beta}{2}\right)\right| = O\left(\frac{1}{\sqrt{\rho}}\log(dn/\beta)\right).
\end{equation}
Therefore, for all $\tau = \tilde{r},\frac{\tilde{r}}{2},\dots,2^{-dn}$,
\begin{align*}
&\widehat{\mathrm{Bias}}\left(\widetilde{\X},\tilde{\tau}\right)+\widehat{\mathrm{Noise}}\left(\widehat{\mathrm{tr}},\tilde{\tau},\frac{\rho}{2},\frac{\beta}{2}\right)
\\
=&2\cdot \min\left(\widehat{\mathrm{Bias}}\left(\widetilde{\X},\tilde{\tau}\right)+\widehat{\mathrm{Noise}}\left(\widehat{\mathrm{tr}},\tilde{\tau},\frac{\rho}{2},\frac{\beta}{2}\right)\right) + \frac{\tilde{r}^2}{n}\cdot  \left|\mathrm{Diff}\left(\widetilde{\X},\widehat{\mathrm{tr}},\tilde{\tau},\frac{\rho}{2},\frac{\beta}{2}\right)\right|
\\
\leq & 2\cdot \max\left(\widehat{\mathrm{Bias}}\left(\widetilde{\X},\tau\right)+\widehat{\mathrm{Noise}}\left(\widehat{\mathrm{tr}},\tau,\frac{\rho}{2},\frac{\beta}{2}\right)\right) + \frac{\tilde{r}^2}{n}\cdot  \left|\mathrm{Diff}\left(\widetilde{\X},\widehat{\mathrm{tr}},\tilde{\tau},\frac{\rho}{2},\frac{\beta}{2}\right)\right|
\\
= & 2\cdot \max\left(\widehat{\mathrm{Bias}}\left(\widetilde{\X},\tau\right)+\widehat{\mathrm{Noise}}\left(\widehat{\mathrm{tr}},\tau,\frac{\rho}{2},\frac{\beta}{2}\right)\right) + \frac{\tilde{r}^2}{n}\cdot  \left|\mathrm{Diff}\left(\widetilde{\X},\widehat{\mathrm{tr}},\tau,\frac{\rho}{2},\frac{\beta}{2}\right)\right|+ O\left(\frac{\tilde{r}^2}{\sqrt{\rho} n}\log(dn/\beta)\right)
\\
= & O\left(\widehat{\mathrm{Bias}}\left(\widetilde{\X},\tau\right)+\widehat{\mathrm{Noise}}\left(\widehat{\mathrm{tr}},\tau,\frac{\rho}{2},\frac{\beta}{2}\right)+\frac{\tilde{r}^2}{\sqrt{\rho}n}\log(dn/\beta)\right),
\end{align*}
where the fourth line is by (\ref{eq:th:err_ada_cov_7}).
Together with (\ref{eq:th:err_ada_cov_6}), finally, we have
\begin{align}
\nonumber
&\|\widetilde{\mathbf{\Sigma}}_{\mathrm{Ada}}-\mathbf{\Sigma}[\widetilde{\X}]\|_F \\
=& O\left(\min_{2^{-dn}\leq \tau\leq \tilde{r}}\left(\widehat{\mathrm{Bias}}\left(\widetilde{\X},\tau\right)+\widehat{\mathrm{Noise}}\left(\widehat{\mathrm{tr}},\tau,\frac{\rho}{2},\frac{\beta}{2}\right)\right)+\frac{\widetilde{r}^2}{\sqrt{\rho}n}\log(dn/\beta)\right)
\\
=& O\left(\min_{2^{-dn}\leq \tau\leq 1}\left(\widehat{\mathrm{Bias}}\left(\widetilde{\X},\tau\right)+\widehat{\mathrm{Noise}}\left(\widehat{\mathrm{tr}},\tau,\frac{\rho}{2},\frac{\beta}{2}\right)\right)+\frac{\widetilde{r}^2}{\sqrt{\rho}n}\log(dn/\beta)\right)
\label{eq:th:err_ada_cov_8}
\end{align}

For the second case, where $\tilde{t} \geq -dn$, by Lemma~\ref{lm:err_SVT}, with probability at least $1-\frac{\beta}{4}$, we have,
\[\mathrm{Diff}\left(\widetilde{\X},\widehat{\mathrm{tr}},2^{\tilde{t}},\frac{\rho}{2},\frac{\beta}{2}\right)\geq -\frac{6\sqrt{2}}{\sqrt{\rho}}\cdot\tilde{r}^2\cdot\log(8(dn+1)/\beta),\]
and 
\[\mathrm{Diff}\left(\widetilde{\X},\widehat{\mathrm{tr}},2^{\tilde{t}+1},\frac{\rho}{2},\frac{\beta}{2}\right)\leq \frac{6\sqrt{2}}{\sqrt{\rho}}\cdot\tilde{r}^2\cdot\log(8(dn+1)/\beta).\]
Recall $\mathrm{Diff}\left(\widetilde{\X},\widehat{\mathrm{tr}},\tau,\frac{\rho}{2},\frac{\beta}{2}\right)$ increases with the decrease of $\tau$, thus we have, for all $\tau = \tilde{r},\frac{\tilde{r}}{2},\dots,2^{-dn}$,
\begin{equation*}
\left| \mathrm{Diff}\left(\widetilde{\X},\widehat{\mathrm{tr}},2^{\tilde{t}+1},\frac{\rho}{2},\frac{\beta}{2}\right) \right|- \left|\mathrm{Diff}\left(\widetilde{\X},\widehat{\mathrm{tr}},\tau,\frac{\rho}{2},\frac{\beta}{2}\right)\right| = O\left(\frac{1}{\sqrt{\rho}}\log(dn/\beta)\right),
\end{equation*}
or 
\begin{equation*}
\left| \mathrm{Diff}\left(\widetilde{\X},\widehat{\mathrm{tr}},2^{\tilde{t}},\frac{\rho}{2},\frac{\beta}{2}\right) \right|- \left|\mathrm{Diff}\left(\widetilde{\X},\widehat{\mathrm{tr}},\tau,\frac{\rho}{2},\frac{\beta}{2}\right)\right| = O\left(\frac{1}{\sqrt{\rho}}\log(dn/\beta)\right).
\end{equation*}
Further using
\[\widehat{\mathrm{Bias}}\left(\widetilde{\X},2^{\tilde{t}+1}\right)+\widehat{\mathrm{Noise}}\left(\widehat{\mathrm{tr}},2^{\tilde{t}+1},\frac{\rho}{2},\frac{\beta}{2}\right) \leq \widehat{\mathrm{Bias}}\left(\widetilde{\X},2^{\tilde{t}}\right)+\widehat{\mathrm{Noise}}\left(\widehat{\mathrm{tr}},2^{\tilde{t}},\frac{\rho}{2},\frac{\beta}{2}\right),\]
and following the same procedure as before, we can also get (\ref{eq:th:err_ada_cov_8}).

Now, we are ready to show how the target error bound can be derived based on (\ref{eq:th:err_ada_cov_8}). Given a $\tau$, recall in (\ref{eq:bias_3}) we have 
\begin{equation}
\label{eq:th:err_ada_cov_9}
\widehat{\mathrm{Bias}}\left(\widetilde{\X},\tau\right) \leq 2\cdot \gamma(\widetilde{\X},\tau).
\end{equation}
Also, recall the definitions of $\mathrm{GaussNoise}(\cdot)$, $\mathrm{SeparateNoise}(\cdot)$ and $\mathrm{Noise}(\cdot)$; combining them with (\ref{eq:th:err_ada_cov_3}) and the fact $\widetilde{\mathrm{tr}}\leq \mathrm{tr}$, we have
\begin{align}
\nonumber
&\widehat{\mathrm{Noise}}\left(\widehat{\mathrm{tr}},\tau,\frac{\rho}{2},\frac{\beta}{2}\right)
\\
\nonumber
=& O\left(\min\left(\frac{\tau^2}{\sqrt{\rho} n}\cdot\omega\left(d,\frac{\beta}{2}\right), \frac{\tau}{\rho^{1/4}\sqrt{n}}\cdot\left(\sqrt{\mathrm{tr}}+\frac{\tilde{r}\cdot\log(1/\beta)^{1/4}}{\rho^{1/4}\sqrt{n}} \right)\cdot\sqrt{\upsilon\left(d,\frac{\beta}{4}\right)} +\frac{\tau^2}{\rho^{1/2}n}\cdot\eta\left(d,\frac{\beta}{4}\right)\right)\right)
\\
\nonumber
=& O\left(\min\left(\frac{\tau^2}{\sqrt{\rho} n}\cdot\omega\left(d,\frac{\beta}{2}\right), \frac{\tau}{\rho^{1/4}\sqrt{n}}\cdot\left(\sqrt{\mathrm{tr}}+\frac{\mathrm{rad}(\X)\cdot\log(1/\beta)^{1/4}}{\rho^{1/4}\sqrt{n}} \right)\cdot\sqrt{\upsilon\left(d,\frac{\beta}{4}\right)} +\frac{\tau^2}{\rho^{1/2}n}\cdot\eta\left(d,\frac{\beta}{4}\right)\right)\right.
\\
\nonumber
&\left.+\frac{2^{-2dn}}{\sqrt{\rho} n}\cdot \sqrt{\upsilon\left(d,\frac{\beta}{4}\right)}\cdot\log\left(dn/\beta\right)\right)
\\
\nonumber
=& O\left(\min\left(\frac{\tau^2}{\sqrt{\rho} n}\cdot\omega\left(d,\frac{\beta}{2}\right), \frac{\tau}{\rho^{1/4}\sqrt{n}}\cdot\left(\sqrt{\mathrm{tr}}+\frac{\sqrt{\mathrm{tr}}\cdot\log(1/\beta)^{1/4}}{\rho^{1/4}} \right)\cdot\sqrt{\upsilon\left(d,\frac{\beta}{4}\right)} +\frac{\tau^2}{\rho^{1/2}n}\cdot\eta\left(d,\frac{\beta}{4}\right)\right)\right.
\\
\nonumber
&\left.+\frac{2^{-2dn}}{\sqrt{\rho} n}\cdot \sqrt{\upsilon\left(d,\frac{\beta}{4}\right)}\cdot\log\left(dn/\beta\right)\right)
\\
\label{eq:th:err_ada_cov_10}
=&O\left(\mathrm{Noise}\left(\X,\tau,\frac{\rho}{2},\frac{\beta}{2}\right)\cdot \frac{\log(1/\beta)^{1/4}}{\rho^{1/4}}+2^{-dn}\right).
\end{align}
The second equality is from (\ref{eq:th:err_ada_cov_1}). The third equality is by the fact $\mathrm{tr}\geq \frac{\mathrm{rad}^2}{n}$, and the last equality is by (\ref{eq:error_detail}).

Finally, combining 
(\ref{eq:th:err_ada_cov_2}), (\ref{eq:th:err_ada_cov_8}), (\ref{eq:th:err_ada_cov_9}) and (\ref{eq:th:err_ada_cov_10}), we have
\begin{align*}
\|\widetilde{\mathbf{\Sigma}}_{\mathrm{Ada}}-\mathbf{\Sigma}\|_F = & O\left(\min_{2^{-dn}\leq \tau\leq \tilde{r}}\left(\mathrm{Noise}\left(\X,\tau,\frac{\rho}{2},\frac{\beta}{2}\right)\cdot \frac{\log(1/\beta)^{1/4}}{\rho^{1/4}}+\gamma(\widetilde{\X},\tau)\right)+2^{-dn}+\frac{\widetilde{r}^2}{\sqrt{\rho}n}\log(dn/\beta)\right)
\\
= & O\left(\min_{2^{-dn}\leq \tau}\left(\mathrm{Noise}\left(\X,\tau,\frac{\rho}{2},\frac{\beta}{2}\right)\cdot \frac{\log(1/\beta)^{1/4}}{\rho^{1/4}}+\gamma(\widetilde{\X},\tau)\right)+2^{-dn}+\frac{\widetilde{r}^2}{\sqrt{\rho}n}\log(dn/\beta)\right)
\\
= & O\left(\min_{2^{-dn}\leq \tau}\left(\mathrm{Noise}\left(\X,\tau,\frac{\rho}{2},\frac{\beta}{2}\right)\cdot \frac{\log(1/\beta)^{1/4}}{\rho^{1/4}}+\gamma(\widetilde{\X},\tau)\right)+2^{-dn}+\frac{\mathrm{rad}(\X)^2}{\sqrt{\rho}n}\log(dn/\beta)\right)
\\
= & O\left(\min_{2^{-dn}\leq \tau}\left(\mathrm{Noise}\left(\X,\tau,\frac{\rho}{2},\frac{\beta}{2}\right)\cdot \frac{\log(1/\beta)^{1/4}}{\rho^{1/4}}\cdot\log(n)+\gamma(\X,\tau)\cdot \frac{\log(dn/\beta)}{\sqrt{\rho}}\right)+2^{-dn}\right)
\\
= & O\left(\min_{\tau}\left(\mathrm{Noise}\left(\X,\tau,\frac{\rho}{2},\frac{\beta}{2}\right)\cdot \frac{\log(1/\beta)^{1/4}}{\rho^{1/4}}\cdot\log(n)+\gamma(\X,\tau)\cdot \frac{\log(dn/\beta)}{\sqrt{\rho}}\right)+2^{-dn}\right).
\end{align*}
The second line follows from the observation that increasing $\tau$ beyond $\tilde{r}$ will increase $\mathrm{Noise}\left(\X,\tau,\frac{\rho}{2},\frac{\beta}{2}\right)$ but will not reduce $\gamma(\widetilde{\X},\tau)$. The third line is by (\ref{eq:th:err_ada_cov_1}). The fourth line is due to $\gamma(\X,\tau)\geq \gamma(\widetilde{\X},\tau)$ and the following fact: if $\tau\geq\mathrm{rad}(\X)$, then $\frac{\tau^2}{n}\geq \frac{\mathrm{rad}(\X)^2}{n}$ and $\frac{\tau\sqrt{\mathrm{tr}}}{\sqrt{n}}\geq \frac{\mathrm{rad}(\X)^2}{n}$; otherwise, $\gamma(\X,\tau)\geq \frac{\mathrm{rad}(\X)^2}{n}$. The last line is because
\[2^{-dn} = \Omega\left(\mathrm{Noise}\left(\X,2^{-dn},\frac{\rho}{2},\frac{\beta}{2}\right)\cdot \frac{\log(1/\beta)^{1/4}}{\rho^{1/4}}\cdot\log(n)\right),\]
and for any $\tau<2^{-nd}$
\[\gamma(\X,\tau)\geq \gamma(\X,2^{-dn}).\] 
\end{proof}

\section{Extension to pure-DP}
\label{sec:pure_dp}

Our proposed algorithms can be extended to the  $\varepsilon$-DP setting, by simply replacing Gaussian noise with Laplace noise. However, the error is enlarged since the added noise is proportional to $\ell_1$ sensitivities, which can be much larger than $\ell_2$ sensitivities. Using the fact that $\|\x\|_2\leq 1$ implies $\|\x\|_1\leq \sqrt{d}$, we immediately obtain an upperbound for the $\ell_1$ sensitivity of computing $\mathbf{\Sigma}$ based on the $\ell_2$ sensitivity:

\begin{lemma}
\label{lm:lap_sens}
For any $\X,\X'\in\mathcal{B}_d^n$, $\X\sim\X'$,
\[\|\mathbf{\Sigma}(\X) - \mathbf{\Sigma}(\X')\|_1 \leq \frac{\sqrt{2}d}{n}.\]
\end{lemma}

The $\ell_1$ sensitivity for computing $\mathbf{\Lambda}$ was shown by~\cite{amin2019differentially} to be:

\begin{lemma}[\cite{amin2019differentially}]
\label{lm:eigensensl1}
For any $\X,\X'\in\mathcal{B}_d^n$, $\X\sim\X'$, 
\begin{equation}
\nonumber
\|\mathbf{\Lambda}-\mathbf{\Lambda}'\|_{1}\leq \frac{2}{n}.
\end{equation}
\end{lemma}

Under $\varepsilon$-DP, our base mechanism is the Laplace mechanism. In the same spirit as the Gaussian mechanism, we draw a symmetric $d\times d$ noise matrix $\W$ where the entries on and above the diagonal are i.i.d. $\mathrm{Lap}(1)$. For convenience, we refer to $\W$ as a Symmetric Laplace Wigner matrix and write $\W\sim \mathrm{SLW}(d)$. Then, we scale the noise matrix by the factor $\frac{\sqrt{2}d}{\varepsilon n}$. We refer to the modified algorithm as $\mathrm{LapCov}$. For $\mathrm{SeparateCov}$, we use $\mathrm{LapCov}$ to replace $\mathrm{GaussCov}$ for privatizing eigenvectors and add Laplace noise in place of Gaussian noise to privatize the eigenvalues.

\subsection{Concentration inequalities}
Before we proceed to anaylze the error of our modified algorithms, some concentration bounds for the Laplace noise vectors and matrices are in order. Our derivations are based on two lemmata.

\begin{lemma} [\cite{sambale2020some}]
    \label{lm:subexp_vec_l2_inq}
Let $Y_1,...,Y_m$ be independent sub-exponential random variables with $\mathbf{E}[Y_k]=0$ for $1\leq k \leq m$. Let $f:\mathbb{R}^m \rightarrow \mathbb{R}$ be convex and $1$-Lipschitz. Then for $Y=(Y_1,...,Y_m)$, and some constant $c>0$,
\begin{equation}
\nonumber
\Pr[|f(Y)-\mathbf{E}[f(Y)]| > t] \leq 2\exp\left(-c\cdot\frac{t}{\log(m)}\right).
\end{equation}
\end{lemma}

\begin{lemma} [\cite{bandeira2016sharp}]
\label{lm:matrix_opnorm}
For any $d>0$, let $\W$ be a $d\times d$ Symmetric Laplace Wigner matrix . Then,
\[\mathbf{E}[\|\W\|_{2}] \leq 3\sqrt{d}+\frac{9\log(d)}{\sqrt{\log(1.5)}}.\]
\end{lemma}

Based on these, we have
\begin{lemma}
\label{lm:lap_vec_inq}
For any $d>0$, $\beta>0$, let $Y=(Y_1,...,Y_d)$ where $Y_k \sim \mathrm{Lap}(1)$, for $1\leq k \leq d$. Then for some constant $c>0$, with probability at least $1-\beta$,
\[\|Y\|_2 \leq \frac{3}{2}\sqrt{d} + O\left(\log(1/\beta)\cdot\log(d)\right).\]
\end{lemma}
\begin{proof}
First, note that $\mathbf{E}\left[\|Y\|_2\right] \leq \frac{3}{2}\sqrt{d}$. Indeed,
\begin{equation}
\nonumber
\frac{\mathbf{E}\left[\|Y\|_2\right]}{\sqrt{d}} = \mathbf{E}\left[\frac{\|Y\|_2}{\sqrt{d}}\right]\leq \mathbf{E}\left[\frac{1}{2}\left(1+\frac{\|Y\|_2^2}{d}\right)\right] = \frac{1}{2}\left(1+\frac{\sum_{k}\mathbf{E}\left[Y_k^2\right]}{d}\right) = \frac{1}{2}\left(1+\frac{\sum_{k}2}{d}\right) = \frac{3}{2},
\end{equation}
where the inequality is due to $\sqrt{z} \leq \frac{1+z}{2}$ for $z\geq 0$ and the second last equality is due to $Y_k \sim \mathrm{Lap}(1)$ with mean $0$ and variance $2$. Applying Lemma \ref{lm:subexp_vec_l2_inq} with $m=d$ then gives the stated inequality.
\end{proof}

\begin{lemma}
\label{lm:matrix_opnorm_bound}
For any $d>0$, $\beta>0$, given $\W\sim \mathrm{SLW}(d)$, then, with probability at least $1-\beta$,
\begin{equation}
\nonumber
\|\W\|_2 \leq 3\sqrt{d}+O\left(\log(1/\beta)\cdot\log(d)\right).
\end{equation}
\end{lemma}

\begin{proof}
Note that the function of computing $\|\cdot\|_{2}$ is convex and $1$-Lipschitz for matrices \cite{wainwright2019high}. Note that $\W$ has $m=d(d+1)/2$ independent entries. We apply Lemma \ref{lm:subexp_vec_l2_inq} to get the stated inequality.
\end{proof}

\begin{lemma}
\label{lm:lap_fnorm_bound}
For any $d>0$, $\beta>0$, given $\W\sim \mathrm{SLW}(d)$, then, with probability at least $1-\beta$,
\begin{equation}
\nonumber
\|\W\|_F \leq \frac{3}{2}d+O\left(\log(1/\beta)\cdot\log(d)\right).
\end{equation}
\end{lemma}

\begin{proof}
Similar to the proof of Lemma \ref{lm:lap_vec_inq}, we compute $\mathbf{E}\left[\|\W\|_F\right]$ and apply Lemma \ref{lm:subexp_vec_l2_inq} with $m=d(d+1)/2$.
\begin{align}
\nonumber
\frac{\mathbf{E}\left[\|\W\|_F\right]}{d} &\leq \frac{1}{2}\cdot\mathbf{E}\left[1+\frac{\|\W\|_F^2}{d^2}\right] 
\\
\nonumber
&= \frac{1}{2}\left(1+\frac{\sum_{k}\mathbf{E}\left[W_{k,k}^2\right]+2\sum_{j<k}\mathbf{E}\left[W_{j,k}^2\right]}{d^2}\right) 
\\
\nonumber
&= \frac{1}{2}\left(1+\frac{2d+2d(d-1)}{d^2}\right) = \frac{3}{2}.
\end{align}
\end{proof}

\subsection{Error bounds}
\label{sec:error_pure_dp}

Based on Lemma~\ref{lm:lap_fnorm_bound}, $\mathrm{LapCov}$ achieves error $O\left(\frac{d^2}{\varepsilon n}+\frac{d}{\varepsilon n}\cdot\log(1/\beta)\cdot\log(d)\right) = \tilde{O}\left(\frac{d^2}{n}\right)$ and by Lemma \ref{lm:lap_vec_inq} and \ref{lm:matrix_opnorm_bound}, we have

\begin{lemma}
\label{lm:err_lap_sep_cov}
Given any $\varepsilon > 0$, for any $\X\in \mathcal{B}_d^n$, $\mathrm{SeparateCov}$ preserves $\varepsilon$-DP and for any $\beta>0$, with probability at least $1-\beta$, $\mathrm{SeparateCov}$ returns a $\widetilde{\mathbf{\Sigma}}_{\mathrm{Lap}}$ such that
\begin{align*}
\|\widetilde{\mathbf{\Sigma}}_{\mathrm{Sep}}-\mathbf{\Sigma}\|_F =& O\left(\frac{\sqrt{d\cdot\mathrm{tr}}}{\sqrt{\varepsilon n} }\cdot \sqrt{\sqrt{d} +\log(1/\beta)\cdot\log(d)}+\frac{1}{\varepsilon n}\left(\sqrt{d}+\log(1/\beta)\cdot \log(d)\right)\right)
\\
=&\tilde{O}\left(\frac{d^{3/4}\sqrt{\mathrm{tr}}}{\sqrt{n}}+\frac{\sqrt{d}}{n}\right).
\end{align*}
\end{lemma}

Under $\varepsilon$-DP, $\mathrm{EMCov}$ achieves the error $\tilde{O}\left(\frac{d\sqrt{\mathrm{tr}}}{\sqrt{n}}+\frac{\sqrt{d}}{n}\right)$, so $\mathrm{SeparateCov}$ improves the $d$ dependency in the first term from $d$ to $d^{3/4}$.  Note that our trace-sensitive bound also implies a worst-case bound of $\tilde{O}({d^{3/4} \over \sqrt{n}})$.  Concurrently with our work, Nikolov~\cite{nikolov2022private}  shows how to achieve a better worst-case bound of $\tilde{O}(\frac{\sqrt{d}}{\sqrt{n}})$, as an application of a new private query release mechanism based on the Johnson-Lindenstrauss transform, although his approach does not seem to yield a trace-sensitive bound.

Finally, we modify the $\mathrm{AdaptiveCov}$ algorithm to support $\varepsilon$-DP. To this end, we first replace $\mathrm{GaussCov}$ with $\mathrm{LapCov}$ and $\mathrm{SeparateCov}$ with its $\varepsilon$-DP version. Then, since $\mathrm{PrivRadius}$ and $\mathrm{SVT}$ already satisfy pure DP, the remaining modification is to use a Laplace tail bound in place of the Gaussian tail bound when constructing the DP upper bound of $\mathrm{tr}$. Based on these, we have

\begin{theorem}
\label{th:err_lap_ada_cov}
Given any $\varepsilon>0$, for any $\X\in\mathcal{B}_d^n$, $\mathrm{AdaptiveCov}$ preserves $\varepsilon$-DP and for any $\beta>0$, with probability at least $1-\beta$, $\mathrm{AdaptiveCov}$ returns a $\widetilde{\mathbf{\Sigma}}_{\mathrm{Ada}}$ such that
\begin{align*}
\left\|\widetilde{\mathbf{\Sigma}}_{\mathrm{Ada}}-\mathbf{\Sigma}\right\|_F =& \tilde{O}\left(\min_{\tau} \left(\min\left(\frac{\tau^2d^2}{n},\frac{\tau d^{3/4}\sqrt{\mathrm{tr}}}{\sqrt{n}} +\frac{\tau^2\sqrt{d}}{n}\right)+\gamma(\X,\tau)\right)+2^{-dn}\right).
\end{align*}
\end{theorem}

\section{Experiments}
\label{sec:experiments}

We conducted experiments\footnote{The code can be found at \url{https://github.com/hkustDB/PrivateCovariance}.} to evaluate our algorithms on both synthetic and real-world datasets. We compare $\mathrm{SeparateCov}$ and $\mathrm{AdaptiveCov}$ against  $\mathrm{GaussCov}$~\cite{dwork2014analyze}, $\mathrm{EMCov}$~\cite{amin2019differentially}.  We implemented $\mathrm{EMCov}$ in Python following the pseudo-code provided in \cite{amin2019differentially} and the descriptions of the sampling algorithm in \cite{kent2018new}.  We also tested $\mathrm{CoinPress}$~\cite{biswas2020coinpress}, but since it is designed to minimize the Mahalanobis error, it does not perform well when measured in Frobenius error. The two distance measures coincide when $\mathbf{\Sigma}$ is well-conditioned but in this case, $\mathrm{CoinPress}$ degenerates into $\mathrm{GaussCov}$. Therefore, we omit it from the reported results.  As a baseline, we include returning a zero matrix, which has error $O(\tr)$, hence a trivial trace-sensitive algorithm.  When $\mathrm{rad}(\X)$ is much smaller than $1$, it is unfair for $\mathrm{GaussCov}$ and $\mathrm{EMCov}$, so we scale all datasets such that $0.5\leq\mathrm{rad}(\X)\leq1$.  As a result, we do not need the step to obtain a private radius in $\mathrm{AdaptiveCov}$, either. Each experiment is repeated 50 times, and we report the average error. Furthermore, we have also conducted experiments under pure-DP and the results can be found in the end.

\subsection{Synthetic Datasets}
\label{sec:synthetic_data}
We generate synthetic datasets by first following the method in \cite{amin2019differentially}, to obtain a matrix $\X=\mathbf{ZU}$, where $\mathbf{U}\in \mathbb{R}^{d\times d}$ is sampled from $U(0,1)$, and $\mathbf{Z}\in \mathbb{R}^{n\times d}$ is sampled from $\mathcal{N}(0,\mathbf{I})$. Then the vectors in $\X$ are adjusted to be centred at $0$ and their norms scaled. In \cite{amin2019differentially}, the vectors are scaled to have unit $\ell_2$ norm; in our experiments, to better control $\tr$ and data skewness, we scale the norms so that they follow the Zipf's law.  More precisely, we divide the norms into $N$ bins. The number  of vectors in the $k$-th bin is proportional to $1/k^s$ and their norm is $2^{k-N}$. The parameter $s$ characterizes the skewness, which we fix as $s=3$. Note that $N=1$ corresponds to the unit-norm case with $\tr=1$.

The results on $\tr=1$ case are shown in Figure \ref{fig:err_d_N1}, which correspond to the worst-case bounds.  The $\rho$ here is fixed at $0.1$ and we examine the error growth w.r.t. $d$ for $n=1000,4000,16000$. The results generally agree with the theory: For low $d$, $\mathrm{GaussCov}$ is (slightly) better than $\mathrm{SeparateCov}$, while the latter is much better for high $d$.  $\mathrm{AdaptiveCov}$  is able to choose the better of the two adaptively, with a small cost due to allocating some privacy budget to estimate $\tr$.  Actually, if $\mathrm{AdaptiveCov}$ is given the precondition that all norms are $1$, this extra cost can be saved.
\begin{figure}[htbp]
         \centering
         \begin{subfigure}[t]{0.3\textwidth}
            \centering
            \includegraphics[width=\textwidth]{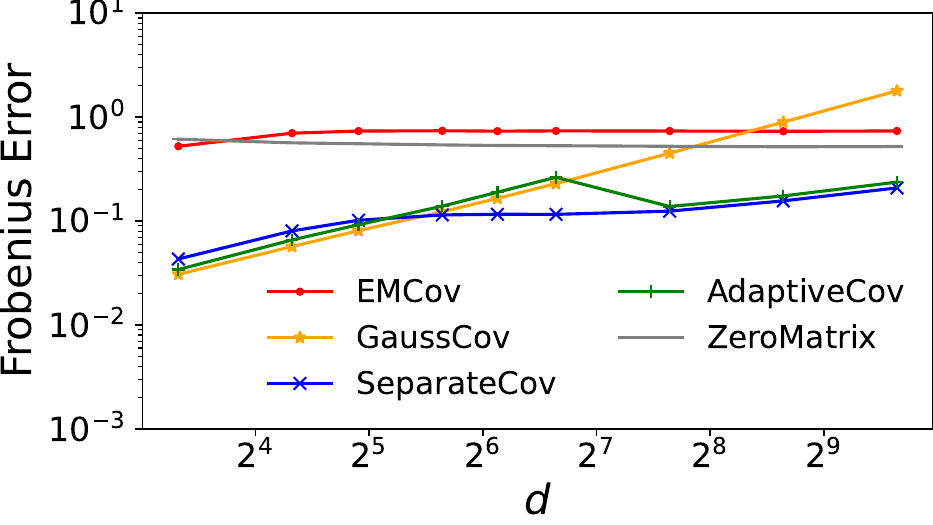}
            \vskip -.08in
            \subcaption{$n=1000$.}
         \end{subfigure}
         \hfill
         \begin{subfigure}[t]{0.3\textwidth}
            \centering
            \includegraphics[width=\textwidth]{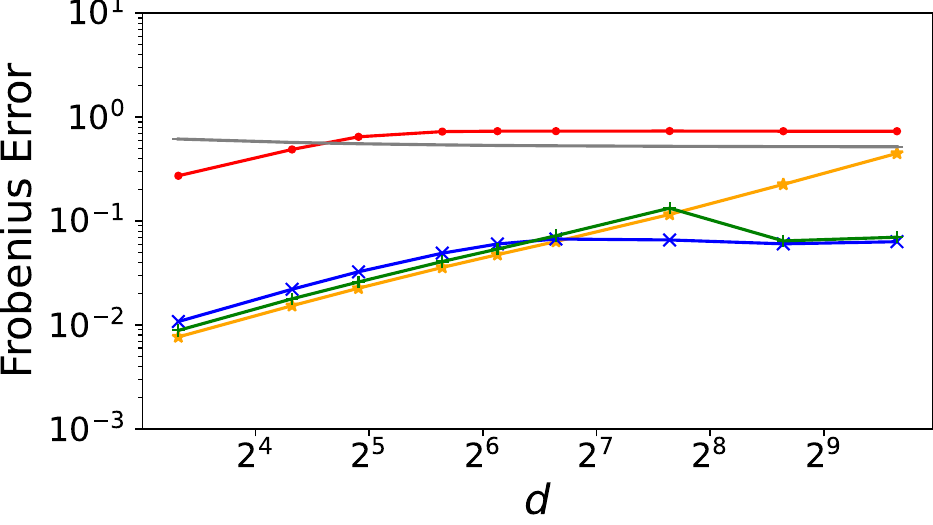}
            \vskip -.08in
            \subcaption{$n=4000$.}
         \end{subfigure}
         \hfill
         \begin{subfigure}[t]{0.3\textwidth}
            \centering
            \includegraphics[width=\textwidth]{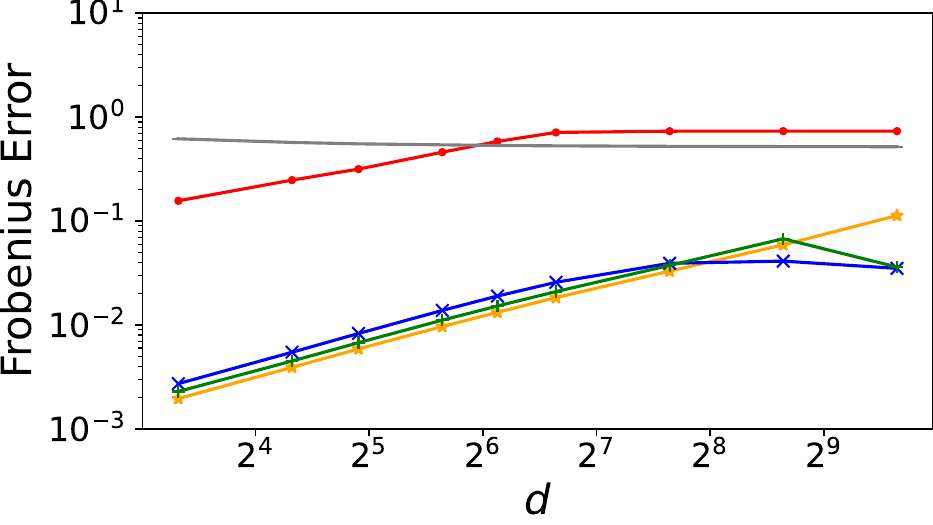}
            \vskip -.08in
            \subcaption{$n=16000$.}
         \end{subfigure}
         \vskip -.01in
         \caption{Results on synthetic datasets fixing $\mathrm{tr}=1$. }
         \label{fig:err_d_N1}
    \vskip -.01in
    \end{figure}
Next, we vary one of the parameters while fixing the others at their default values $d=200$, $n=50000$, $N=4$ and $\rho=0.1$, and the results are reported in Figure \ref{fig:err_d_n_rho_N}.  The most interesting case is Figure \ref{fig:err_d_n_rho_N}(a), where we increase $N$, hence reducing $\tr$, which demonstrates the trace-sensitive bounds.  Clearly, $\mathrm{GaussCov}$ is not trace-sensitive, while the other 4 methods are.  In fact, returning the zero matrix is the best trace-sensitive algorithm if $\tr$ is sufficiently small.  However, this may not be very meaningful in practice, as $N=2^5$ means that most data have norm  $2^{-31}$ but a few have norm $1$.  These few may be outliers and should be removed anyway.   Figure \ref{fig:err_d_n_rho_N}(b)--(d) shows that higher $d$, smaller $n$, and smaller $\rho$ all have similar effects, i.e., $\mathrm{SeparateCov}$ becomes better while $\mathrm{GaussCov}$ becomes worse, while $\mathrm{AdaptiveCov}$ is able to pick the better one. 
\begin{figure}[htbp]
         \centering
         \begin{subfigure}[t]{0.22\textwidth}
             \centering
             \includegraphics[width=\textwidth]{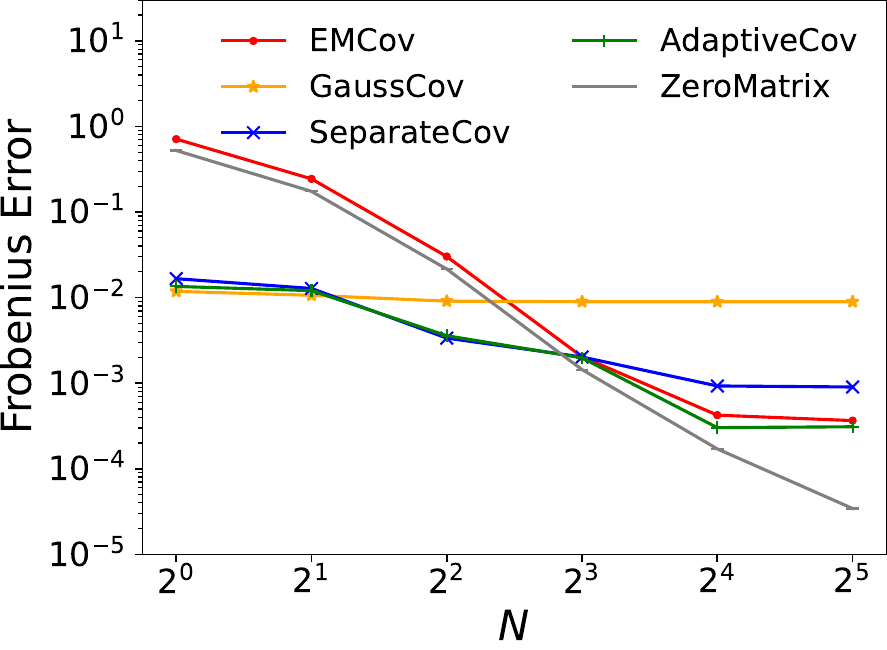}
             \vskip -.08in
             \subcaption{various $\mathrm{tr}$.}
         \end{subfigure}
         \hfill
         \begin{subfigure}[t]{0.22\textwidth}
            \centering
            \includegraphics[width=\textwidth]{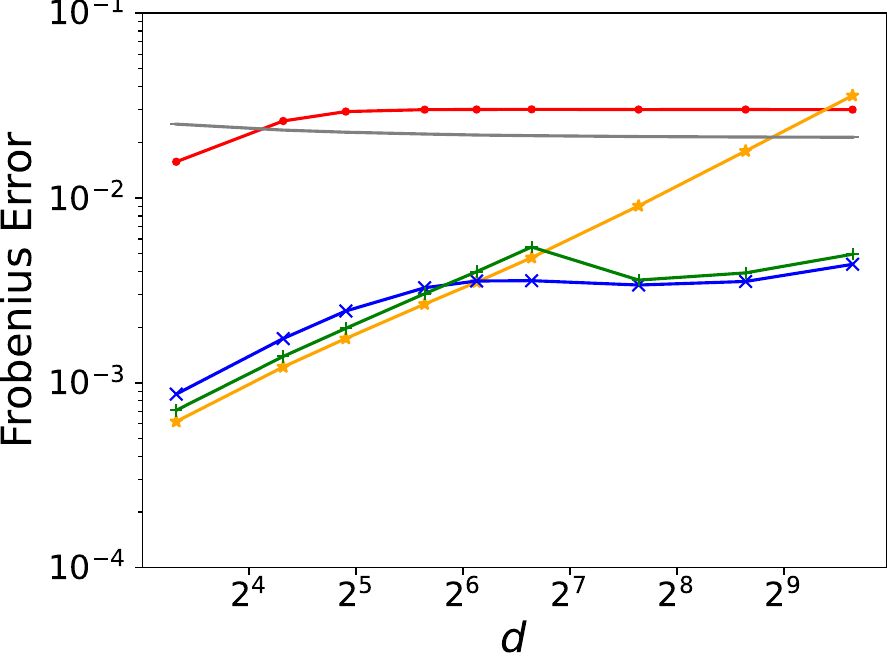}
            \vskip -.08in
            \subcaption{various $d$.}
         \end{subfigure}
         \hfill
         \begin{subfigure}[t]{0.22\textwidth}
            \centering
            \includegraphics[width=\textwidth]{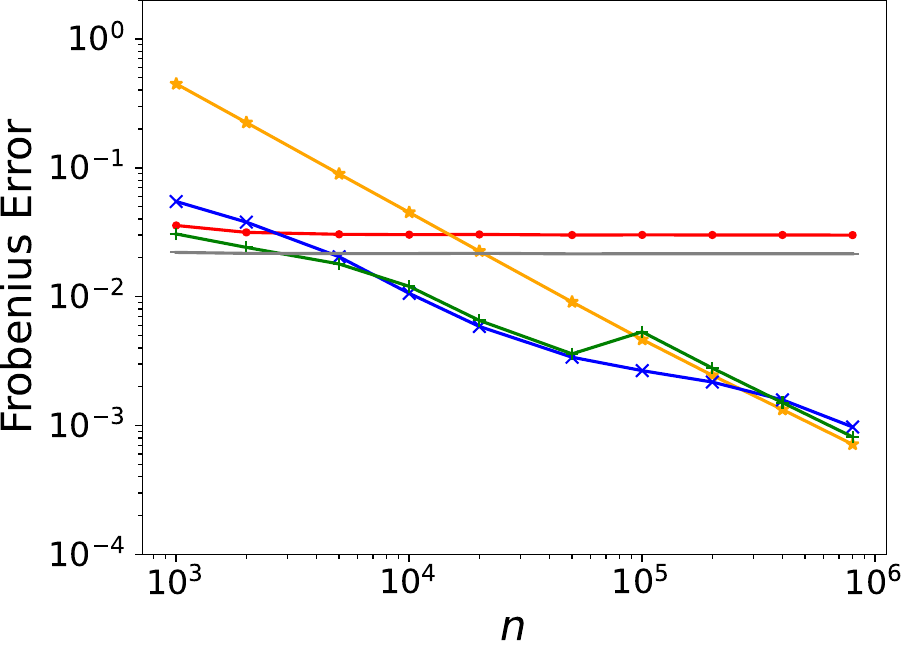}
            \vskip -.08in
            \subcaption{various $n$.}
         \end{subfigure}
         \hfill
         \begin{subfigure}[t]{0.22\textwidth}
            \centering
            \includegraphics[width=\textwidth]{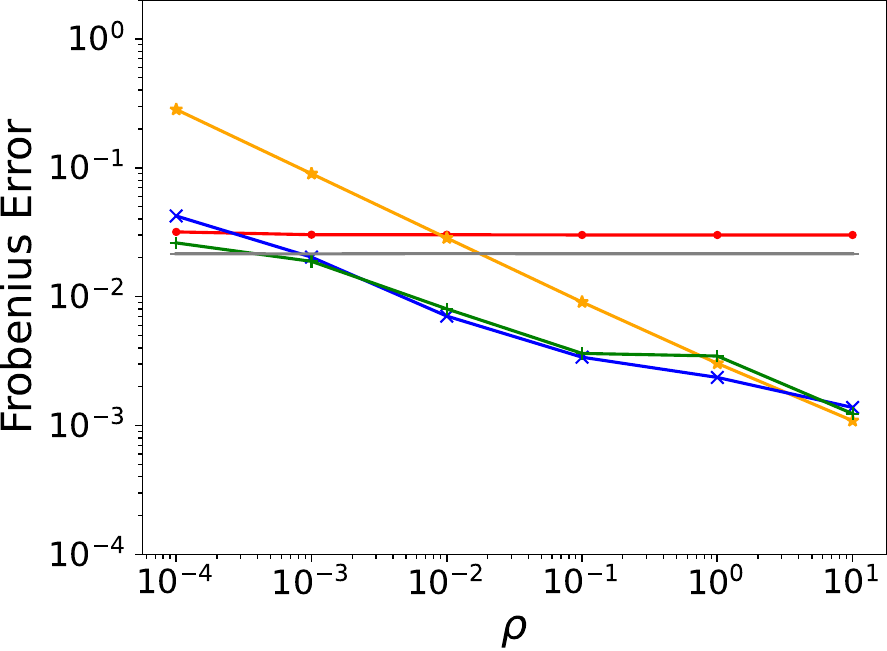}
            \vskip -.08in
            \subcaption{various $\rho$.}
         \end{subfigure}
         \vskip -.01in
         \caption{Results on synthetic datasets as $d,n,N$ or $\rho$ varies. }
         \label{fig:err_d_n_rho_N}
    \vskip -.01in
    \end{figure}

\subsection{Real-world Datasets}
\label{sec:real_data}
We also evaluate the algorithms on two real-world datasets. The first dataset is the MNIST \cite{lecun1998} dataset, which contains images of handwritten digits. We use its training dataset which contains $60,000$ images represented as vectors in $\mathbb{Z}_{255}^d$, where $d=784=28\times 28$. These vectors are normalized by $255\sqrt{d}$ in the experiments. We estimate $\widetilde{\bSigma}$ using samples containing all the digits (Fig.\ref{fig:err_mnist_rho}a-c), we also estimate $\widetilde{\bSigma}$ corresponding to individual digits (Fig.\ref{fig:err_mnist_rho}d-e). In the first case, $\widetilde{\bSigma}$ can be used for further dimensionality reduction analysis; in the second case, individual $\widetilde{\bSigma}$ can be used for modelling the distributions of individual digits, which together can be used in a collective model for classification (e.g. using a mixture model). The second dataset contains news commentary data \cite{wmt2021} consisting of approximately $15,000$ articles, each containing $500-4300$ words, which we convert into vectors of various dimensions using the hashing trick implemented in the scikit-learn package. In this case, the estimated $\widetilde{\bSigma}$ can be used to help with further feature selection for NLP models, for example. These vectors are normalized to have unit $\ell_2$ norm or normalized by the max $\ell_2$ norm.

The experimental results on these two real dataset are shown in Figure \ref{fig:err_mnist_rho} and \ref{fig:err_news_rho}, where we vary $n,d$, and $\rho$, respectively.  On these results, we see that $\mathrm{GaussCov}$ never outperforms $\mathrm{SeparateCov}$, except for a very small advantage in a few cases where we have $\tr=1$, a low $d$, and a high $\rho$.  Another interesting observation is that $\mathrm{AdaptiveCov}$ outperforms both $\mathrm{GaussCov}$ and $\mathrm{SeparateCov}$ in many cases, something that is not apparent on the synthetic datasets.  We believe that this is because these real datasets have heavier tails than the Zipf distribution (we used $s=3$ for Zipf), which makes the adaptive clipping threshold selection more effective.  This really demonstrates the benefits of a tail-sensitive bound.
\begin{figure}[htbp]
         \centering
         \begin{subfigure}[t]{0.3\textwidth}
            \centering
            \includegraphics[width=\textwidth]{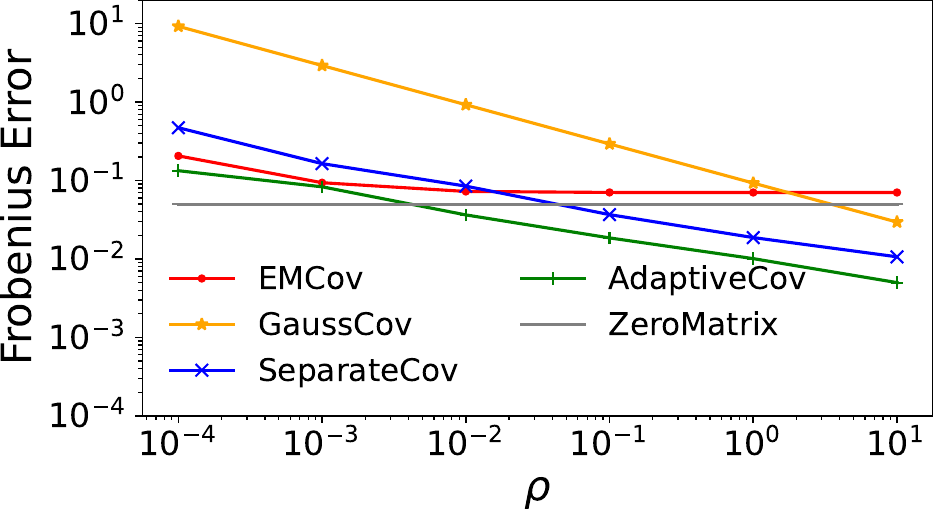}
            \vskip -.08in
            \subcaption{$n=6000,\mathrm{tr}<1$.}
             \vspace{0.08in}
         \end{subfigure}
         \hfill
         \begin{subfigure}[t]{0.3\textwidth}
            \centering
            \includegraphics[width=\textwidth]{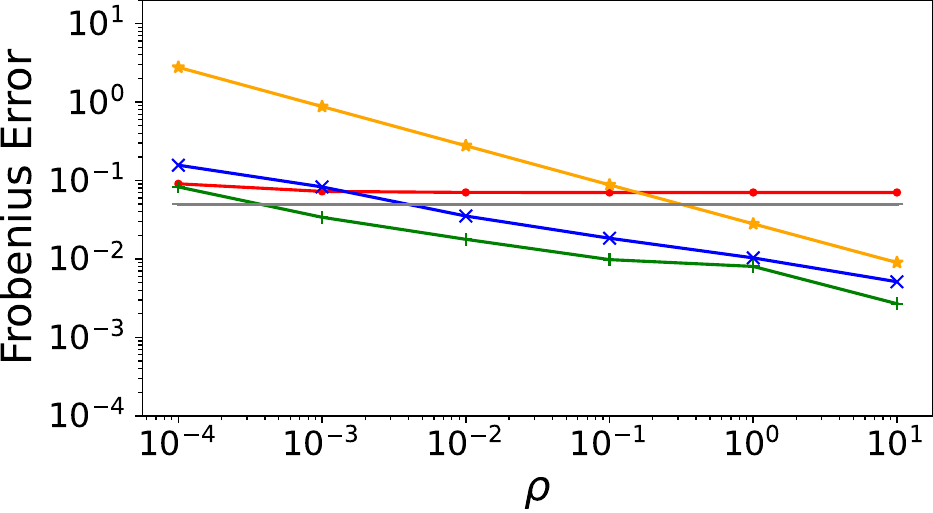}
            \vskip -.08in
            \subcaption{$n=20000,\mathrm{tr}<1$.}
             \vspace{0.08in}
         \end{subfigure}
         \hfill
         \begin{subfigure}[t]{0.3\textwidth}
            \centering
            \includegraphics[width=\textwidth]{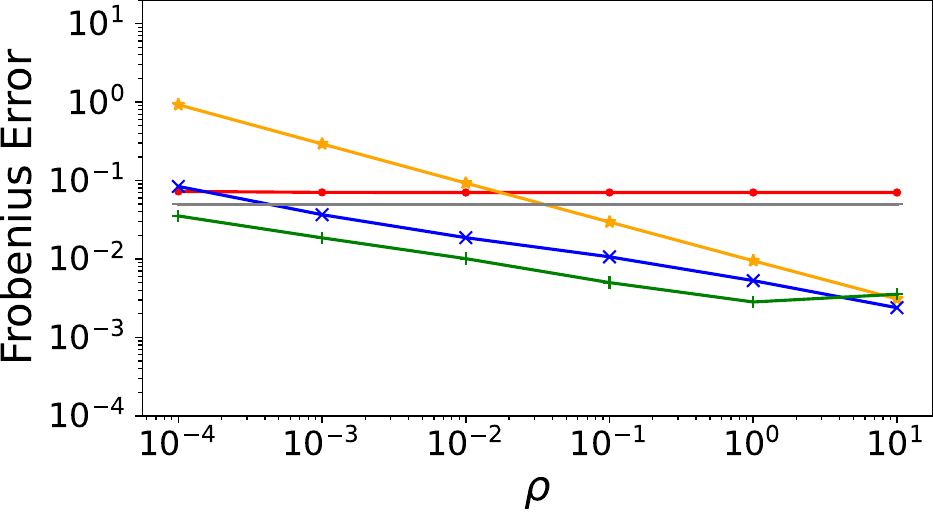}
            \vskip -.08in
            \subcaption{$n=60000,\mathrm{tr}<1$.}
             \vspace{0.08in}
         \end{subfigure}
                  \centering
         \begin{subfigure}[t]{0.3\textwidth}
            \centering
            \includegraphics[width=\textwidth]{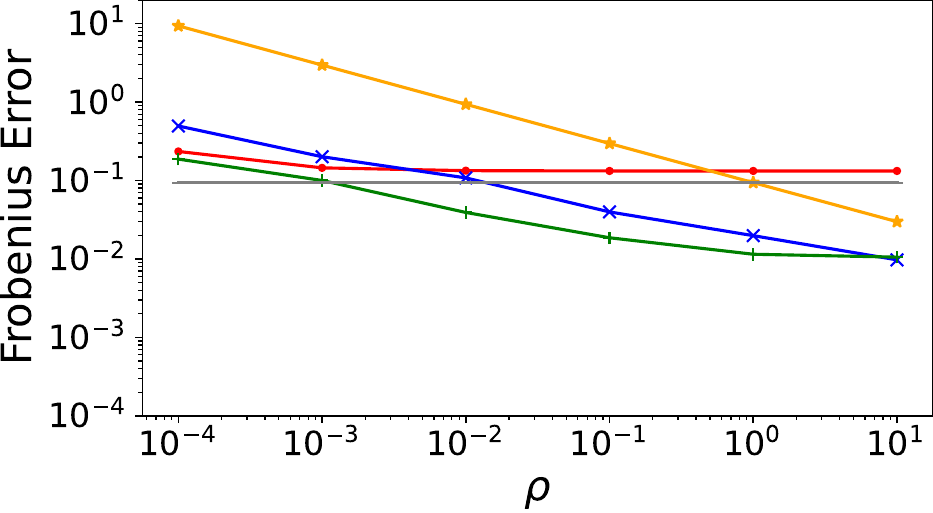}
            \vskip -.08in
            \subcaption{Digit $0$.}
         \end{subfigure}
         \hfill
         \begin{subfigure}[t]{0.3\textwidth}
            \centering
            \includegraphics[width=\textwidth]{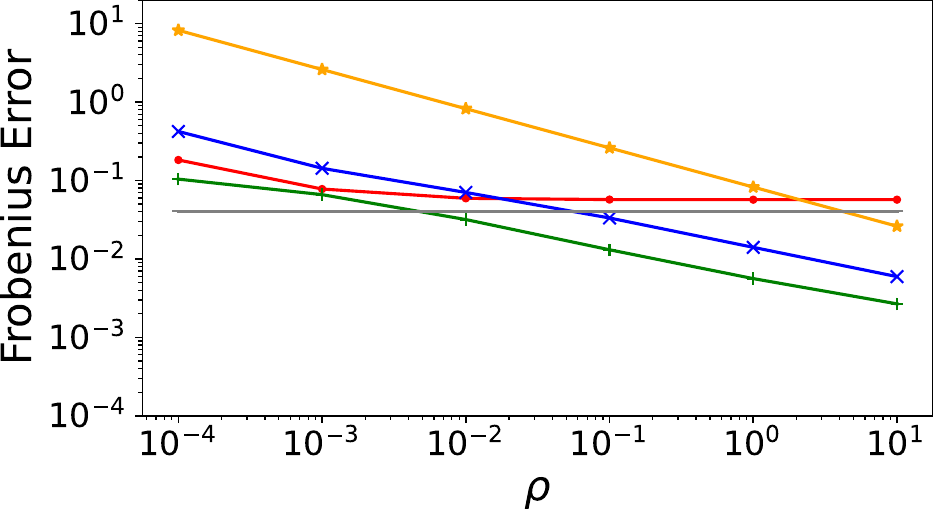}
            \vskip -.08in
            \subcaption{Digit $1$.}
         \end{subfigure}
         \hfill
         \begin{subfigure}[t]{0.3\textwidth}
            \centering
            \includegraphics[width=\textwidth]{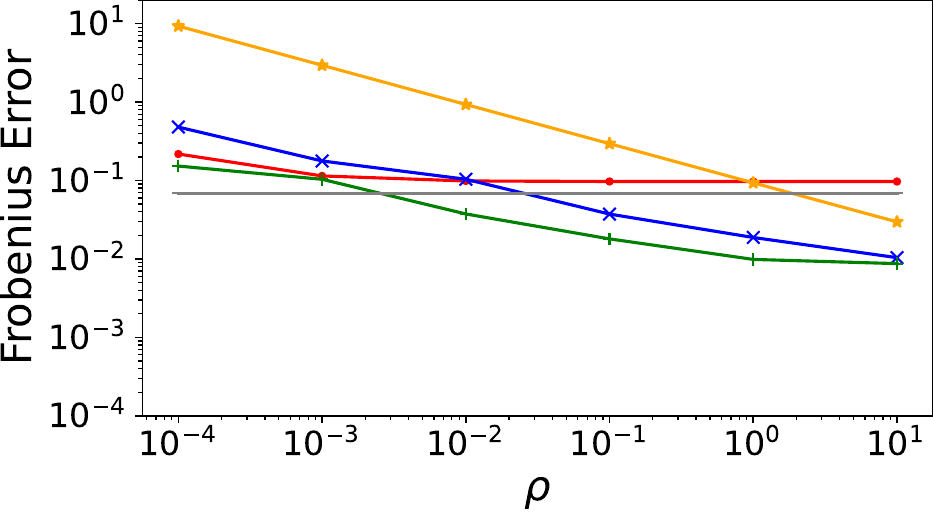}
            \vskip -.08in
            \subcaption{Digit $2$.}
         \end{subfigure}
         \vskip -.01in
         \caption{Results on MNIST dataset. }
         \label{fig:err_mnist_rho}
    \vskip -.01in
    \end{figure}

\begin{figure}[htbp]
    \vskip -.01in
     \centering
         \begin{subfigure}[t]{0.3\textwidth}
            \centering
            \includegraphics[width=\textwidth]{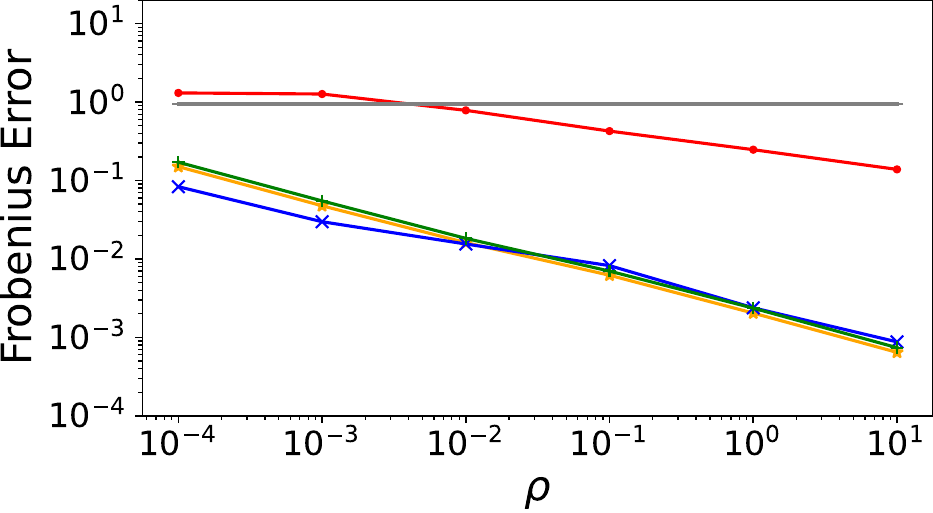}
            \vskip -.08in
            \subcaption{$d=32,\mathrm{tr}=1$.}
            \vspace{0.08in}
         \end{subfigure}
         \hfill
         \begin{subfigure}[t]{0.3\textwidth}
            \centering
            \includegraphics[width=\textwidth]{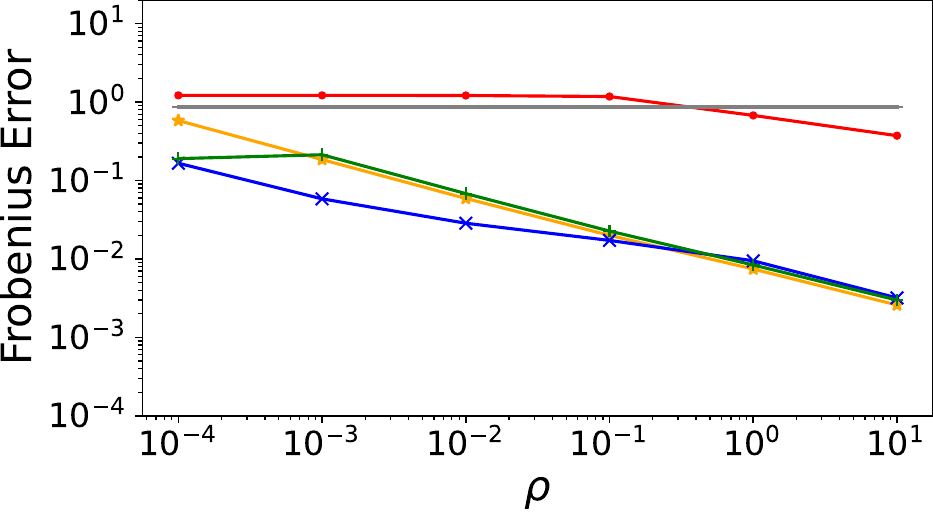}
            \vskip -.08in
            \subcaption{$d=128,\mathrm{tr}=1$.}
            \vspace{0.08in}
         \end{subfigure}
         \hfill
         \begin{subfigure}[t]{0.3\textwidth}
            \centering
            \includegraphics[width=\textwidth]{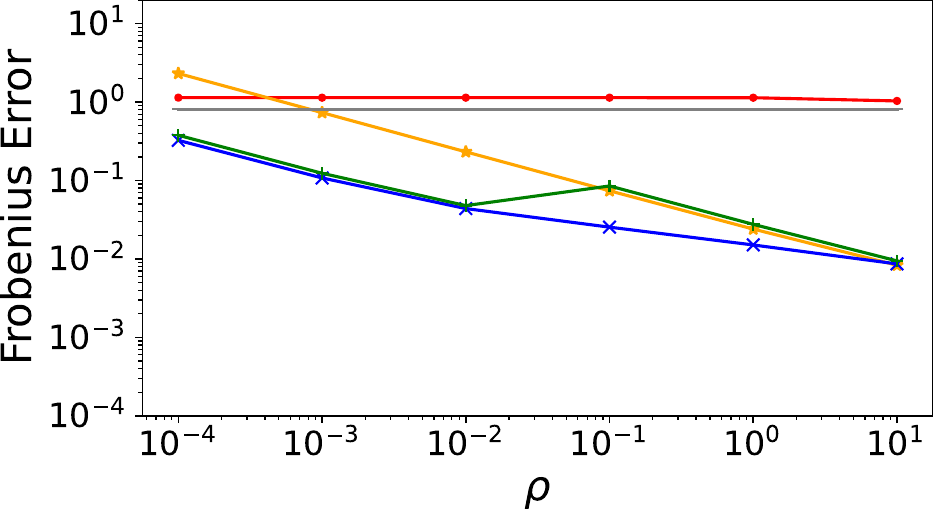}
            \vskip -.08in
            \subcaption{$d=512,\mathrm{tr}=1$.}
            \vspace{0.08in}
         \end{subfigure}
         \centering
         \begin{subfigure}[t]{0.3\textwidth}
            \centering
            \includegraphics[width=\textwidth]{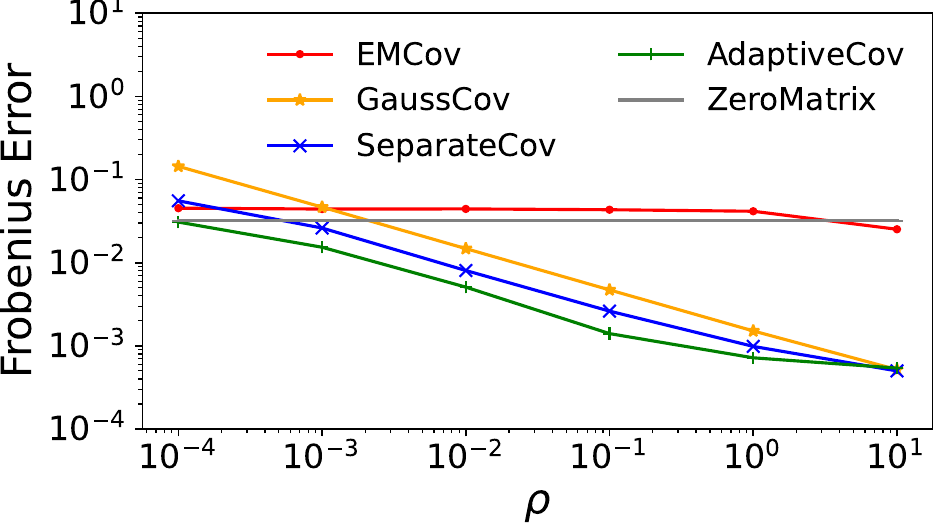}
            \vskip -.08in
            \subcaption{$d=32,\mathrm{tr}<1$.}
         \end{subfigure}
         \hfill
         \begin{subfigure}[t]{0.3\textwidth}
            \centering
            \includegraphics[width=\textwidth]{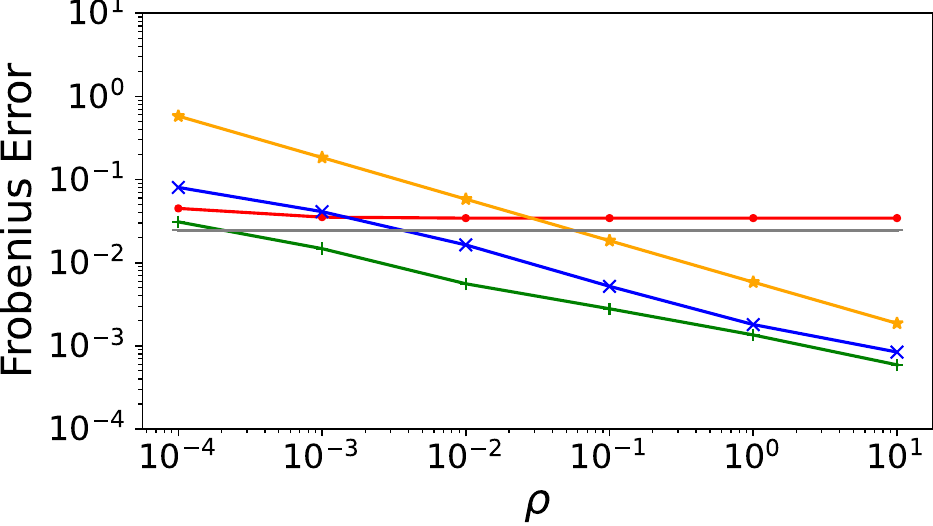}
            \vskip -.08in
            \subcaption{$d=128,\mathrm{tr}<1$.}
         \end{subfigure}
         \hfill
         \begin{subfigure}[t]{0.3\textwidth}
            \centering
            \includegraphics[width=\textwidth]{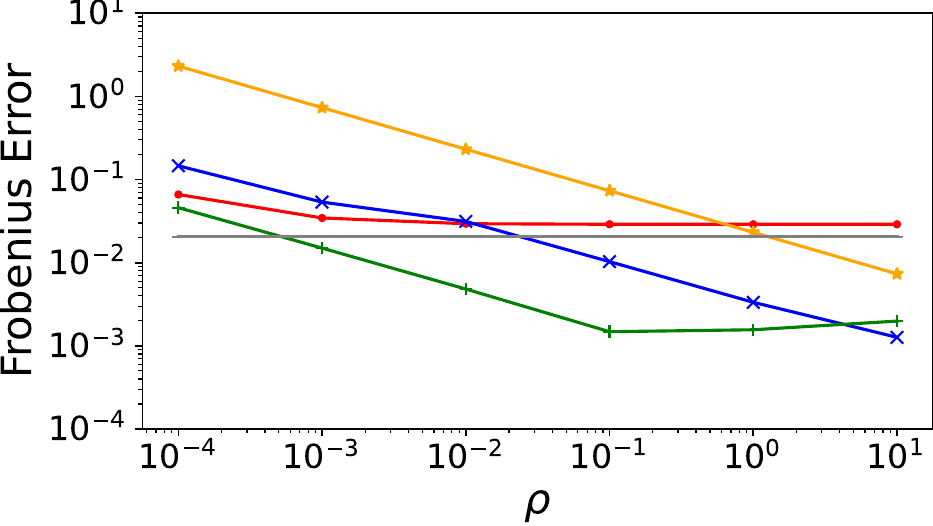}
            \vskip -.08in
            \subcaption{$d=512,\mathrm{tr}<1$.}
         \end{subfigure}
         \vskip -.01in
         \caption{Results on the news dataset. }
         \label{fig:err_news_rho}
    \vskip -.01in
    \end{figure}
    
\subsection{Pure DP experiments}
We also conducted experiments using the pure DP algorithms on the same datasets as in Section~\ref{sec:experiments}. The parameter settings are also mostly the same, except we replaced $\varepsilon$ with $\rho$. For synthetic datasets, we also present an extra set of experiments with a larger sample size $n$. Since the dependency on $d$ degrades for all the algorithms, a larger sample size is needed in order for them to have an advantage over the zero matrix. As mentioned in Appendix~\ref{sec:approximate_DP_nips_2019}, \cite{amin2019differentially} propose two ways to allocate privacy budget for eigenvector estimation under pure-DP; one of them being uniform allocation and the other is to allocate more budget to larger (privatized) eigenvalues. We tested both and their error appeared similar, with the latter being slightly better in some cases (as consistent with their report). Here, we only report the results for the latter budget allocation scheme.

The experimental results agree with our theoretical analysis and are similar as those in $\rho$-zCDP setting:  $\mathrm{SepCov}$ always outperforms $\mathrm{EMCov}$ and will perform better than $\mathrm{LapCov}$ with higher $d$, smaller $n$, and smaller $\varepsilon$. Also, as expected, to outperform the zero matrix, we require larger $n$, smaller $d$, and larger $\varepsilon$ compared with the $\rho$-zCDP setting. For real-world datasets, since they are highly dimensional with limited samples, none of the algorithms have better utility than the zero matrix (except when $d$ is small and $\varepsilon$ is large); as such we do not report the results for real-world datasets here.
\begin{figure}[htbp]
         \centering
         \begin{subfigure}[t]{0.3\textwidth}
            \centering
            \includegraphics[width=\textwidth]{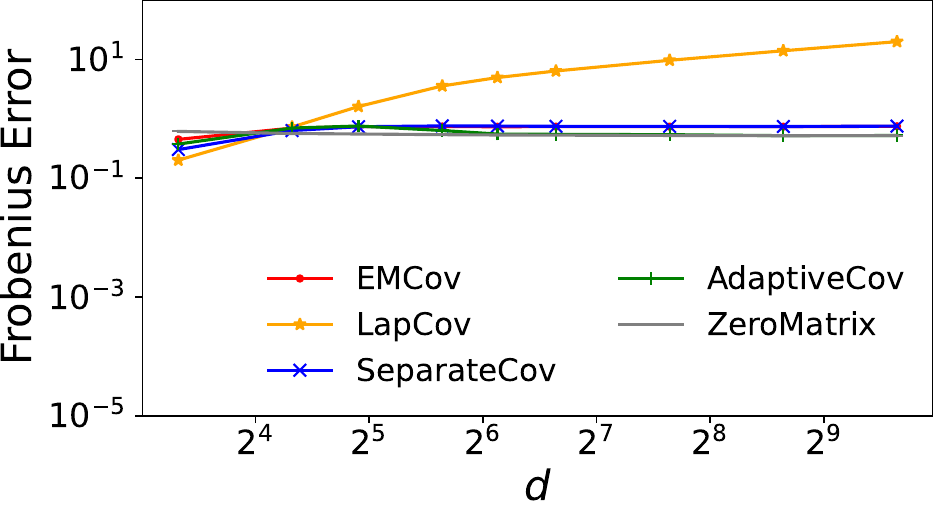}
            \vskip -.08in
            \subcaption{$n=1000$.}
         \end{subfigure}
         \hfill
         \begin{subfigure}[t]{0.3\textwidth}
            \centering
            \includegraphics[width=\textwidth]{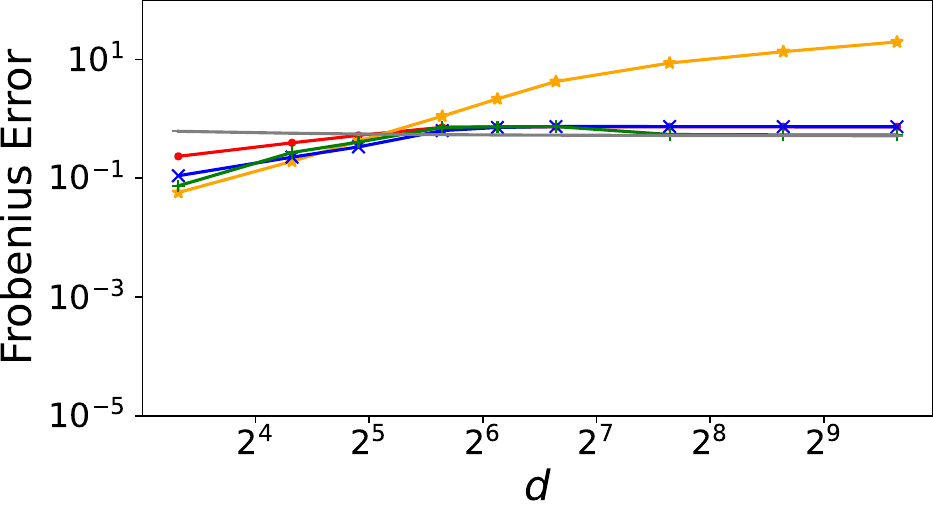}
            \vskip -.08in
            \subcaption{$n=4000$.}
         \end{subfigure}
         \hfill
         \begin{subfigure}[t]{0.3\textwidth}
            \centering
            \includegraphics[width=\textwidth]{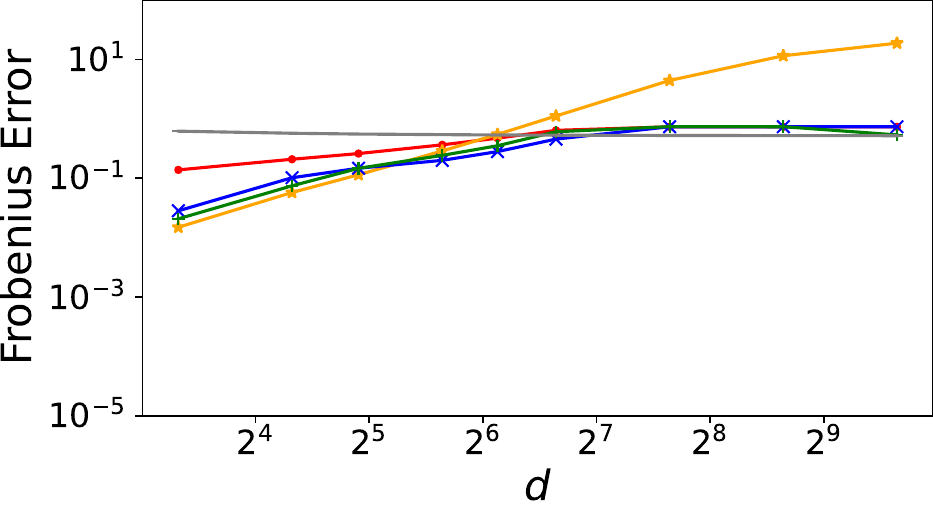}
            \vskip -.08in
            \subcaption{$n=16000$.}
         \end{subfigure}
                 \begin{subfigure}[t]{0.3\textwidth}
            \centering
            \includegraphics[width=\textwidth]{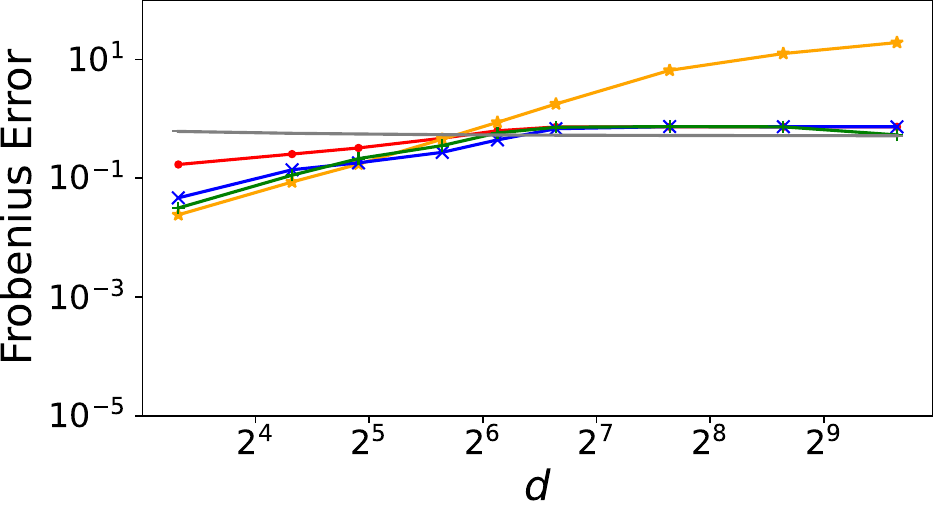}
            \vskip -.08in
            \subcaption{$n=10000$.}
         \end{subfigure}
         \hfill
         \begin{subfigure}[t]{0.3\textwidth}
            \centering
            \includegraphics[width=\textwidth]{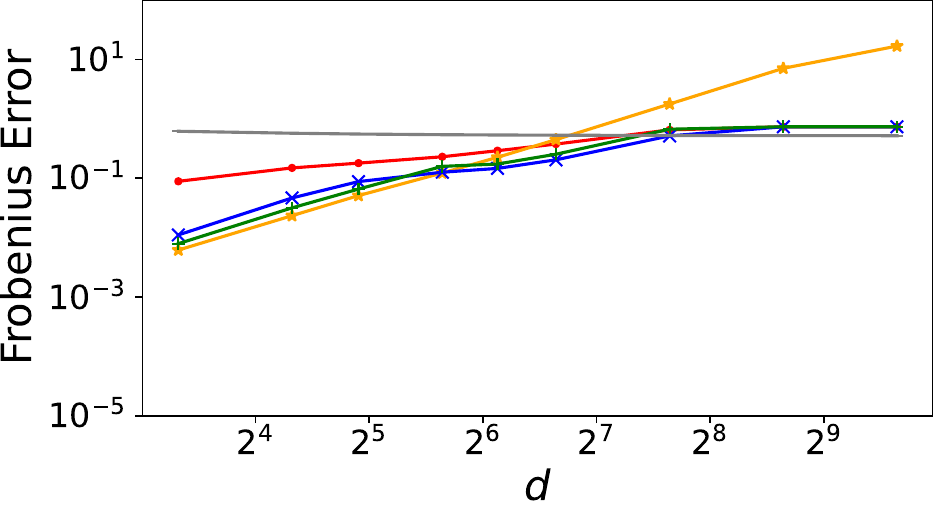}
            \vskip -.08in
            \subcaption{$n=40000$.}
         \end{subfigure}
         \hfill
         \begin{subfigure}[t]{0.3\textwidth}
            \centering
            \includegraphics[width=\textwidth]{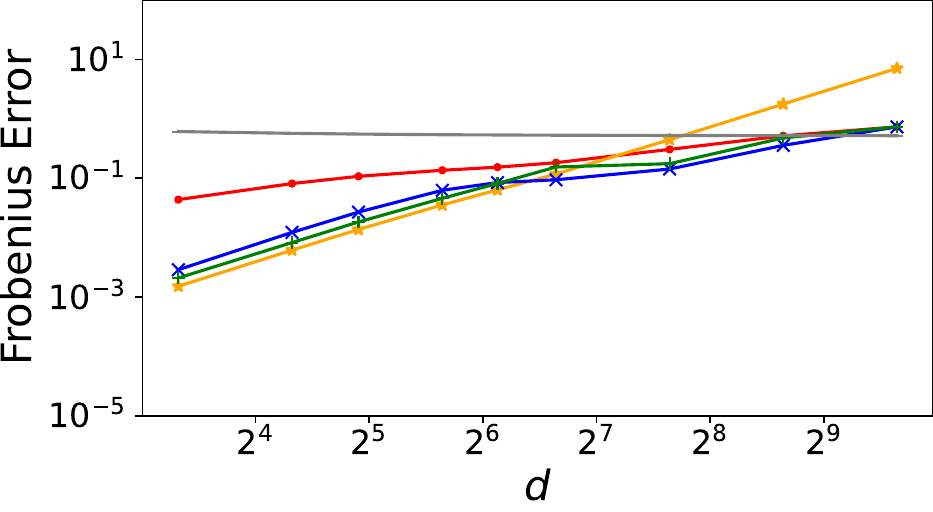}
            \vskip -.08in
            \subcaption{$n=160000$.}
         \end{subfigure}
         \vskip -.01in
         \caption{Results on synthetic datasets fixing $\mathrm{tr}=1$. }
         \label{fig:err_d_N1_pure}
    \vskip -.01in
    \end{figure}

\begin{figure}[htbp]
         \centering
         \begin{subfigure}[t]{0.3\textwidth}
            \centering
            \includegraphics[width=\textwidth]{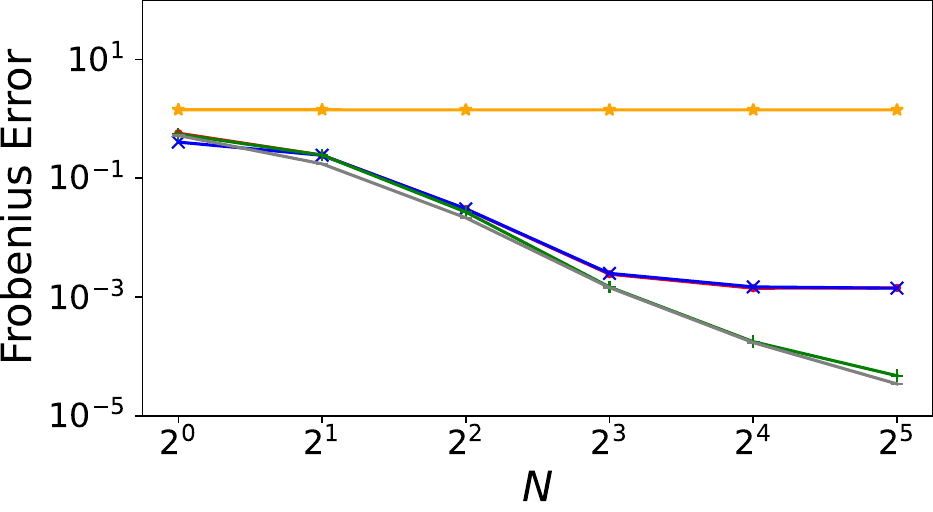}
            \vskip -.08in
            \subcaption{$n=50000$.}
         \end{subfigure}
         \hfill
         \begin{subfigure}[t]{0.3\textwidth}
            \centering
            \includegraphics[width=\textwidth]{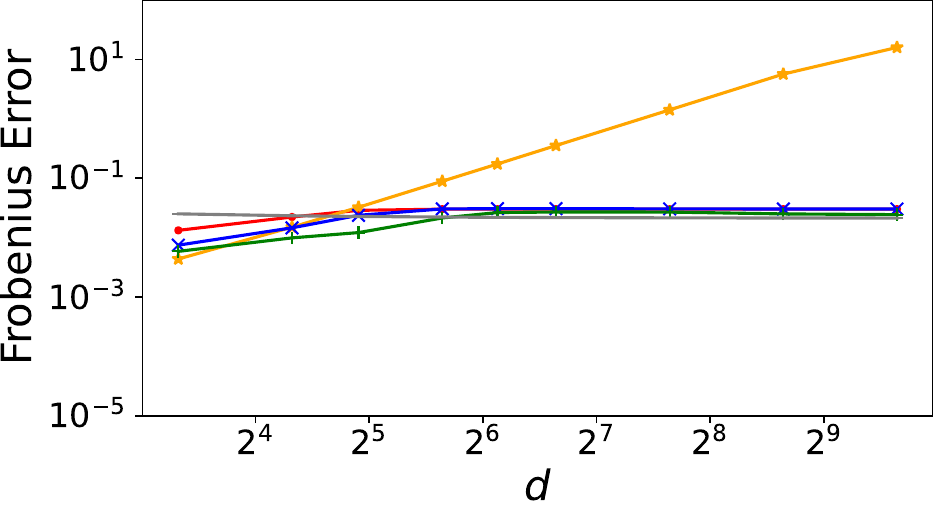}
            \vskip -.08in
            \subcaption{$n=50000$.}
         \end{subfigure}
         \hfill
         \begin{subfigure}[t]{0.3\textwidth}
            \centering
            \includegraphics[width=\textwidth]{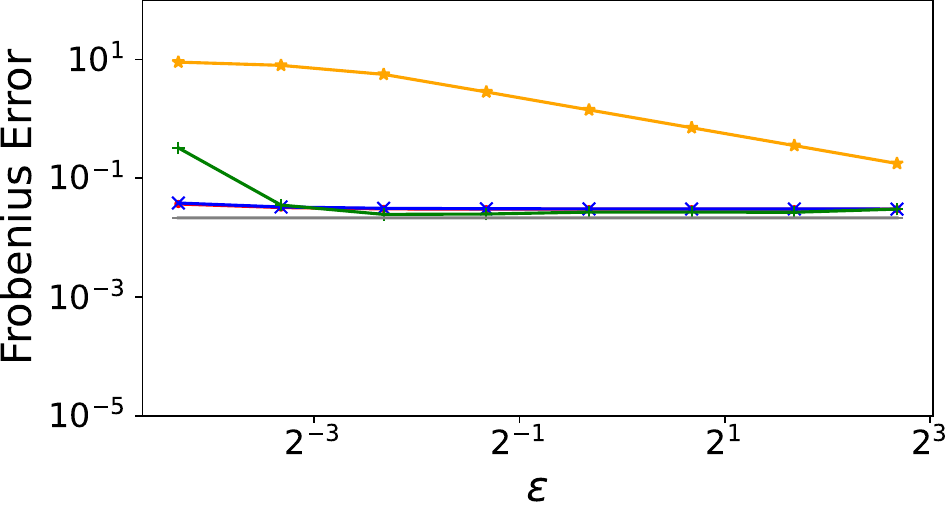}
            \vskip -.08in
            \subcaption{$n=50000$.}
         \end{subfigure}
                 \begin{subfigure}[t]{0.3\textwidth}
            \centering
            \includegraphics[width=\textwidth]{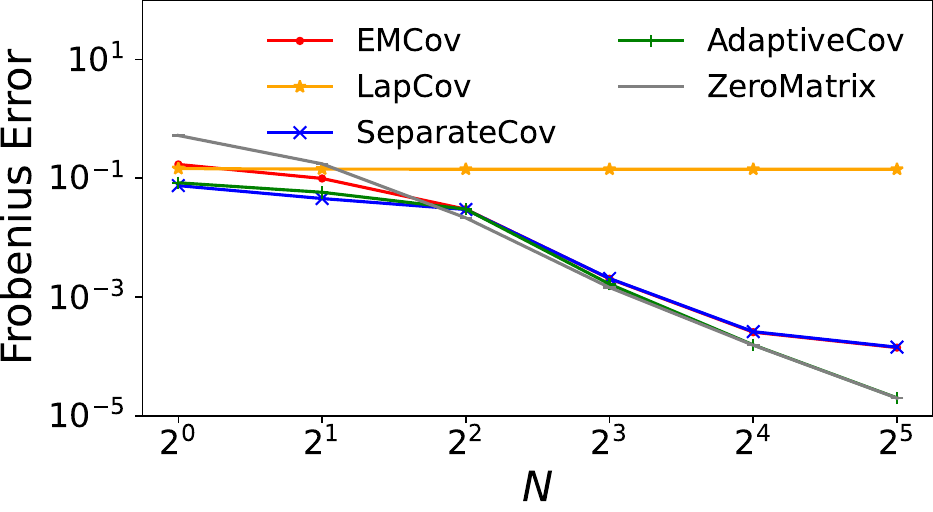}
            \vskip -.08in
            \subcaption{$n=500000$.}
         \end{subfigure}
         \hfill
         \begin{subfigure}[t]{0.3\textwidth}
            \centering
            \includegraphics[width=\textwidth]{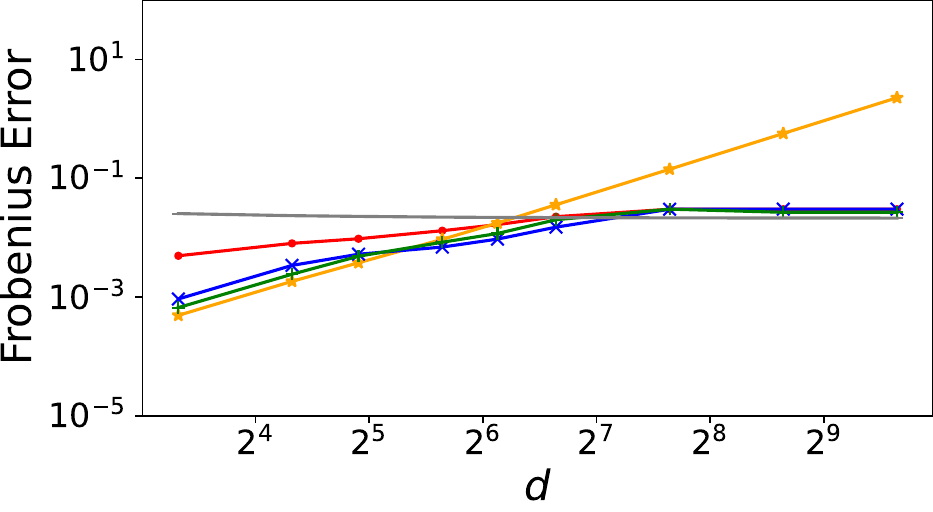}
            \vskip -.08in
            \subcaption{$n=500000$.}
         \end{subfigure}
         \hfill
         \begin{subfigure}[t]{0.3\textwidth}
            \centering
            \includegraphics[width=\textwidth]{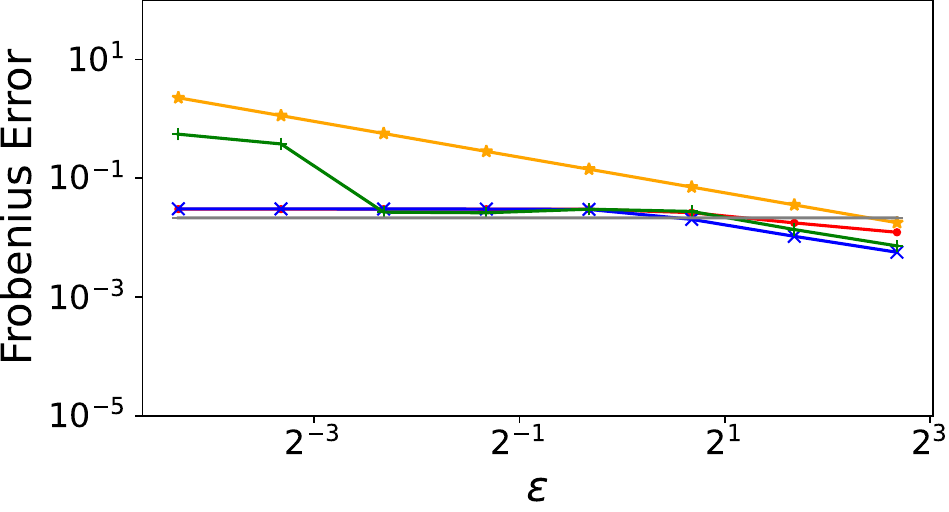}
            \vskip -.08in
            \subcaption{$n=500000$.}
         \end{subfigure}
         \vskip -.01in
         \caption{Results on synthetic datasets as $d,N$ or $\varepsilon$ varies. }
         \label{fig:err_pure}
    \vskip -.01in
    \end{figure}

\section*{Acknowledgements}
This work has been supported by HKRGC under grants 16201819,
16205420, and 16205422.  We would like to thank Aleksandar Nikolov for helpful discussions on the projection mechanism.

\bibliographystyle{alpha}
\bibliography{ref}

\newcommand{\etalchar}[1]{$^{#1}$}
\begin{thebibliography}{KMS{\etalchar{+}}22b}

\bibitem[AAAK21]{aden2021sample}
Ishaq Aden-Ali, Hassan Ashtiani, and Gautam Kamath.
\newblock On the sample complexity of privately learning unbounded
  high-dimensional gaussians.
\newblock In {\em Algorithmic Learning Theory}, pages 185--216. PMLR, 2021.

\bibitem[AD20]{asi2020instance}
Hilal Asi and John~C Duchi.
\newblock Instance-optimality in differential privacy via approximate inverse
  sensitivity mechanisms.
\newblock {\em Advances in neural information processing systems}, 33, 2020.

\bibitem[ADK{\etalchar{+}}19]{amin2019differentially}
Kareem Amin, Travis Dick, Alex Kulesza, Andr{\'e}s~Munoz Medina, and Sergei
  Vassilvitskii.
\newblock Differentially private covariance estimation.
\newblock In {\em NeurIPS}, pages 14190--14199, 2019.

\bibitem[AKMV19]{amin2019bounding}
Kareem Amin, Alex Kulesza, Andres Munoz, and Sergei Vassilvtiskii.
\newblock Bounding user contributions: A bias-variance trade-off in
  differential privacy.
\newblock In {\em International Conference on Machine Learning}, pages
  263--271. PMLR, 2019.

\bibitem[AL22]{ashtiani2022private}
Hassan Ashtiani and Christopher Liaw.
\newblock Private and polynomial time algorithms for learning gaussians and
  beyond.
\newblock In {\em Conference on Learning Theory}, pages 1075--1076. PMLR, 2022.

\bibitem[ATMR21]{andrew2021differentially}
Galen Andrew, Om~Thakkar, Brendan McMahan, and Swaroop Ramaswamy.
\newblock Differentially private learning with adaptive clipping.
\newblock {\em Advances in Neural Information Processing Systems}, 34, 2021.

\bibitem[BDKU20]{biswas2020coinpress}
Sourav Biswas, Yihe Dong, Gautam Kamath, and Jonathan Ullman.
\newblock Coinpress: Practical private mean and covariance estimation.
\newblock {\em Advances in Neural Information Processing Systems}, 33, 2020.

\bibitem[BKSW19]{bun2019private}
Mark Bun, Gautam Kamath, Thomas Steinke, and Steven~Z Wu.
\newblock Private hypothesis selection.
\newblock {\em Advances in Neural Information Processing Systems}, 32, 2019.

\bibitem[BS16]{bun2016concentrated}
Mark Bun and Thomas Steinke.
\newblock Concentrated differential privacy: Simplifications, extensions, and
  lower bounds.
\newblock In {\em Theory of Cryptography Conference}, pages 635--658. Springer,
  2016.

\bibitem[BVH16]{bandeira2016sharp}
Afonso~S Bandeira and Ramon Van~Handel.
\newblock Sharp nonasymptotic bounds on the norm of random matrices with
  independent entries.
\newblock {\em The Annals of Probability}, 44(4):2479--2506, 2016.

\bibitem[CSS13]{chaudhuri2013near}
Kamalika Chaudhuri, Anand~D Sarwate, and Kaushik Sinha.
\newblock A near-optimal algorithm for differentially-private principal
  components.
\newblock {\em Journal of Machine Learning Research}, 14, 2013.

\bibitem[DMNS06]{dwork2006calibrating}
Cynthia Dwork, Frank McSherry, Kobbi Nissim, and Adam Smith.
\newblock Calibrating noise to sensitivity in private data analysis.
\newblock In {\em Theory of cryptography conference}, pages 265--284. Springer,
  2006.

\bibitem[DNR{\etalchar{+}}09]{dwork2009complexity}
Cynthia Dwork, Moni Naor, Omer Reingold, Guy~N Rothblum, and Salil Vadhan.
\newblock On the complexity of differentially private data release: efficient
  algorithms and hardness results.
\newblock In {\em Proceedings of the forty-first annual ACM symposium on Theory
  of computing}, pages 381--390, 2009.

\bibitem[DNT15]{dwork2015efficient}
Cynthia Dwork, Aleksandar Nikolov, and Kunal Talwar.
\newblock Efficient algorithms for privately releasing marginals via convex
  relaxations.
\newblock {\em Discrete \& Computational Geometry}, 53(3):650--673, 2015.

\bibitem[DR14]{dwork2014algorithmic}
Cynthia Dwork and Aaron Roth.
\newblock The algorithmic foundations of differential privacy.
\newblock {\em Foundations and Trends{\textregistered} in Theoretical Computer
  Science}, 9(3--4):211--407, 2014.

\bibitem[DTTZ14]{dwork2014analyze}
Cynthia Dwork, Kunal Talwar, Abhradeep Thakurta, and Li~Zhang.
\newblock Analyze gauss: optimal bounds for privacy-preserving principal
  component analysis.
\newblock In {\em Proceedings of the forty-sixth annual ACM symposium on Theory
  of computing}, pages 11--20, 2014.

\bibitem[DY21]{dong2021universal}
Wei Dong and Ke~Yi.
\newblock Universal private estimators.
\newblock {\em arXiv preprint arXiv:2111.02598}, 2021.

\bibitem[HLY21]{huang2021instance}
Ziyue Huang, Yuting Liang, and Ke~Yi.
\newblock Instance-optimal mean estimation under differential privacy.
\newblock {\em Advances in Neural Information Processing Systems}, 2021.

\bibitem[KGM18]{kent2018new}
John~T Kent, Asaad~M Ganeiber, and Kanti~V Mardia.
\newblock A new unified approach for the simulation of a wide class of
  directional distributions.
\newblock {\em Journal of Computational and Graphical Statistics},
  27(2):291--301, 2018.

\bibitem[KLSU19]{KamathLSU19}
Gautam Kamath, Jerry Li, Vikrant Singhal, and Jonathan Ullman.
\newblock Privately learning high-dimensional distributions.
\newblock In {\em Proceedings of the 32nd Annual Conference on Learning
  Theory}, COLT '19, pages 1853--1902, 2019.

\bibitem[KMS22a]{kamath2022new}
Gautam Kamath, Argyris Mouzakis, and Vikrant Singhal.
\newblock New lower bounds for private estimation and a generalized
  fingerprinting lemma.
\newblock {\em arXiv preprint arXiv:2205.08532}, 2022.

\bibitem[KMS{\etalchar{+}}22b]{kamath2022private}
Gautam Kamath, Argyris Mouzakis, Vikrant Singhal, Thomas Steinke, and Jonathan
  Ullman.
\newblock A private and computationally-efficient estimator for unbounded
  gaussians.
\newblock In {\em Conference on Learning Theory}, pages 544--572. PMLR, 2022.

\bibitem[KMV22]{kothari2022private}
Pravesh Kothari, Pasin Manurangsi, and Ameya Velingker.
\newblock Private robust estimation by stabilizing convex relaxations.
\newblock In {\em Conference on Learning Theory}, pages 723--777. PMLR, 2022.

\bibitem[KRSU10]{kasiviswanathan2010price}
Shiva~Prasad Kasiviswanathan, Mark Rudelson, Adam Smith, and Jonathan Ullman.
\newblock The price of privately releasing contingency tables and the spectra
  of random matrices with correlated rows.
\newblock In {\em Proceedings of the forty-second ACM symposium on Theory of
  computing}, pages 775--784, 2010.

\bibitem[KT13]{kapralov2013differentially}
Michael Kapralov and Kunal Talwar.
\newblock On differentially private low rank approximation.
\newblock In {\em Proceedings of the twenty-fourth annual ACM-SIAM symposium on
  Discrete algorithms}, pages 1395--1414. SIAM, 2013.

\bibitem[LCB98]{lecun1998}
Yann LeCun, Corinna Cortes, and Christopher~J.C. Burges.
\newblock The mnist database of handwritten digits, 1998.
\newblock Available online at: \url{http://yann.lecun.com/exdb/mnist/}. Last
  accessed: May. 2022.

\bibitem[LKO22]{liu2022differential}
Xiyang Liu, Weihao Kong, and Sewoong Oh.
\newblock Differential privacy and robust statistics in high dimensions.
\newblock In {\em Conference on Learning Theory}, pages 1167--1246. PMLR, 2022.

\bibitem[LM00]{laurent2000adaptive}
Beatrice Laurent and Pascal Massart.
\newblock Adaptive estimation of a quadratic functional by model selection.
\newblock {\em Annals of Statistics}, pages 1302--1338, 2000.

\bibitem[MRTZ17]{mcmahan2017learning}
H~Brendan McMahan, Daniel Ramage, Kunal Talwar, and Li~Zhang.
\newblock Learning differentially private recurrent language models.
\newblock {\em arXiv preprint arXiv:1710.06963}, 2017.

\bibitem[Nik22]{nikolov2022private}
Aleksandar Nikolov.
\newblock Private query release via the johnson-lindenstrauss transform.
\newblock {\em arXiv preprint arXiv:2208.07410}, 2022.

\bibitem[NRS07]{nissim2007smooth}
Kobbi Nissim, Sofya Raskhodnikova, and Adam Smith.
\newblock Smooth sensitivity and sampling in private data analysis.
\newblock In {\em Proceedings of the thirty-ninth annual ACM symposium on
  Theory of computing}, pages 75--84, 2007.

\bibitem[NTZ13]{nikolov2013geometry}
Aleksandar Nikolov, Kunal Talwar, and Li~Zhang.
\newblock The geometry of differential privacy: the sparse and approximate
  cases.
\newblock In {\em Proceedings of the forty-fifth annual ACM symposium on Theory
  of computing}, pages 351--360, 2013.

\bibitem[oMTW21]{wmt2021}
Sixth~Conference on~Machine Translation~WMT21.
\newblock News commentary crawl v16, 2021.
\newblock Available online at:
  \url{https://www.statmt.org/wmt20/translation-task.html}. Last accessed: May.
  2022.

\bibitem[PSY{\etalchar{+}}19]{pichapati2019adaclip}
Venkatadheeraj Pichapati, Ananda~Theertha Suresh, Felix~X Yu, Sashank~J Reddi,
  and Sanjiv Kumar.
\newblock Adaclip: Adaptive clipping for private sgd.
\newblock {\em arXiv preprint arXiv:1908.07643}, 2019.

\bibitem[Sam20]{sambale2020some}
Holger Sambale.
\newblock Some notes on concentration for $\alpha$-subexponential random
  variables.
\newblock {\em arXiv preprint arXiv:2002.10761}, 2020.

\bibitem[Smi11]{smith2011privacy}
Adam Smith.
\newblock Privacy-preserving statistical estimation with optimal convergence
  rates.
\newblock In {\em Proceedings of the forty-third annual ACM symposium on Theory
  of computing}, pages 813--822, 2011.

\bibitem[SS21]{singhal2021privately}
Vikrant Singhal and Thomas Steinke.
\newblock Privately learning subspaces.
\newblock {\em Advances in Neural Information Processing Systems},
  34:1312--1324, 2021.

\bibitem[Upa18]{upadhyay2018price}
Jalaj Upadhyay.
\newblock The price of privacy for low-rank factorization.
\newblock In {\em NeurIPS}, 2018.

\bibitem[Wai19]{wainwright2019high}
Martin~J Wainwright.
\newblock {\em High-dimensional statistics: A non-asymptotic viewpoint},
  volume~48.
\newblock Cambridge University Press, 2019.

\end{thebibliography}

\appendix

\section{zCDP error bound for $\mathrm{EMCov}$}
\label{sec:approximate_DP_nips_2019}

The state-of-the-art trace-sensitive algorithm \cite{amin2019differentially} has error $\tilde{O}\left(\sqrt{d \sum_{i=1}^d (\lambda_i/\varepsilon_i)} + \sqrt{d})\right)$.  Under pure-DP, setting $\varepsilon_i=\varepsilon/d$ for all $i$ turns the bound into  $\tilde{O}(d\sqrt{\mathrm{tr}} + \sqrt{d})$.  For zCDP, we can set $\varepsilon_i = \tilde{O}(\sqrt{\rho/d})$, which implies $\frac{\rho}{d}$-zCDP. Then the whole process follows $\rho$-zCDP by composition. Alternatively, one could adaptively set $\varepsilon_i$ according to a privatized $\lambda_i$ as mentioned in \cite{amin2019differentially}, but we do not consider this optimization as it does not yield a better trace-sensitive bound.  Also, as observed from their experiments \cite{amin2019differentially}, this optimization does not have a big effect on the performance of their algorithm. 

The algorithm in \cite{amin2019differentially} estimates $\mathbf{\Sigma}= \X\X^T$.  The ${1\over n}$ normalization factor scales down both the error and $\tr$ by a factor of $n$. This yields the stated bounds in Section \ref{sec:traceintro} and \ref{sec:pure_dp}.

\section{The lower bound of mean estimation from statistical setting to empirical setting}
\label{sec:lower_bound_statistical_to_empirical}

\cite{KamathLSU19} considers the statistical setting, where each $\x_i\sim\mathcal{P}$ and $\mathcal{P}$ is a distribution over $\{-1,1\}^d$ and show the lower bound of $(\varepsilon,\delta)$-DP for $\delta\leq \frac{\varepsilon}{64n}$ is $\Omega(\frac{d}{\varepsilon n})$ (see Lemma 6.2 in that paper).  Since zCDP implies approximate-DP and each instance under statistical setting is naturally a case under empirical setting, their lower bound implies the one of zCDP under empirical setting. Besides, under their setting, data have the $\ell_2$ norm equal to $\sqrt{d}$. Scaling the result by $\sqrt{d}$ and we get the lower bound under our setting.


\section{Applying the projection mechanism to covariance estimation}
\label{sec:projection}

Similar as $\mathrm{SeparateCov}$, the projection mechanism first constructs $\widetilde{\bSigma}_{\mathrm{Gau}}$ by adding Gaussian noise on $\bSigma$. Let $K$ be the set of covariance matrices for all possible datasets $\X$, i.e.,
\[K = \{\bSigma(\X):\X\in(\mathcal{B}_d)^n\}.\]
The projection mechanism return a $\overline{\bSigma}$ such that, 
\begin{equation}
\label{eq:K}
\overline{\bSigma} = \argmin_{\overline{\bSigma}\in K} \|\overline{\bSigma}-\widetilde{\bSigma}_{\mathrm{Gau}}\|_F.
\end{equation}

For the privacy, projection mechanism does some post-processing on $\widetilde{\bSigma}_{\mathrm{Gau}}$ thus naturally preserves DP. For the utility, by Lemma 1 of~\cite{nikolov2013geometry}, $\|\bSigma-\overline{\bSigma}\|_F$ is bounded by maximum eigenvalue of noise matrix added in $\widetilde{\bSigma}_{\mathrm{Gau}}$, which is further bounded by $\tilde{O}(\sqrt{d})$ by Lemma~\ref{lm:upper_bound_noise_SGW}. For the efficiency, after using obtaining $\widetilde{\bSigma}_{\mathrm{Gau}}$, which requires time $O(n+d^2)$, the projection mechanism uses Frank-Wolfe algorithm to solve the optimal projection of $\widetilde{\bSigma}_{\mathrm{Gau}}$ on $K$, where there are $n$ iterations with each one requires computing the SVD of a PSD matrix and the total running time is $O(nd^3)$. Below, we show that optimal projection can be solved by only computing one SVD thus the cost can be reduced to $O(d^3)$.

\begin{lemma}
\label{lm:projection}
Let $\mathbf{B}$ be a symmetric $d\times d$ matrix with $\mathbf{B} = \mathbf{P} \mathbf{U} \mathbf{P}^T$, and $K$ is defined as (\ref{eq:K}). Then there is $\mathbf{A'} \in \mathrm{arg}\min_{\mathbf{A}\in K} \|\mathbf{B}-\mathbf{A}\|_F$ where $\mathbf{A}'=\mathbf{P} \mathbf{U}'\mathbf{P}^T$ for some diagonal matrix $\mathbf{U}'\succeq 0$.
\end{lemma}

\begin{proof}
It suffices to show for any $\mathbf{A}\in K$, we can find a such $\mathbf{A}'$ and $\|\mathbf{B}-\mathbf{A}'\|_F\leq \|\mathbf{B}-\mathbf{A}\|_F$. 

Let $\mathbf{U}' =\mathrm{diag}(\mathbf{P}^T \mathbf{A} \mathbf{P})$, then, $\mathrm{tr}(\mathbf{U}') = \mathrm{tr}(\mathbf{P}^T \mathbf{A} \mathbf{P}) = \mathrm{tr}(\mathbf{A}) \leq 1$. Also, $U'_{i,i}={P}_i^T\mathbf{A}{P}_i\leq \lambda_1(\mathbf{A})\leq 1$ and $U'_{i,i}={P}_i^T\mathbf{A}{P}_i\geq 0$ since $\mathbf{A}$ is PSD. Overall, $\mathbf{A}'\in K$.

Furthermore,
\begin{align*}
\|\mathbf{B}-\mathbf{A}\|_F^2 =& \|\mathbf{P}\mathbf{U} \mathbf{P}^T-\mathbf{A}\|_F^2 
\\
=& \|\mathbf{U}-\mathbf{P}^T\mathbf{A}\mathbf{P}\|_F^2 
\\
=& 2\sum_{i<j} (\mathbf{P}^T\mathbf{A}\mathbf{P})_{i,j}^2 + \sum_{i}\left(U_{i,i}-(\mathbf{P}^T\mathbf{A}\mathbf{P})_{i,i}\right)^2
\\
\geq & \sum_{i}\left(U_{i,i}-(\mathbf{P}^T\mathbf{A}\mathbf{P})_{i,i}\right)^2
\\
= & \sum_{i}\left(U_{i,i}-U_{i,i}'\right)^2
\\
= & \|\mathbf{B}-\mathbf{A}'\|_F^2
\end{align*}
\end{proof}

Let $\widetilde{\bSigma}_{\mathrm{Gau}} = \widetilde{P} \widetilde{\mathbf{\Lambda}} \widetilde{P}^T$ and $\widetilde{\mathbf{\Lambda}} = \mathrm{diag}(\widetilde{\Lambda}_1,\dots,\widetilde{\Lambda}_d)$.  From Lemma \ref{lm:projection}, to find $\mathrm{arg}\min_{\overline{\bSigma}\in K} \|\overline{\bSigma} - \widetilde{\bSigma}_{\mathrm{Gau}}\|_F$, it suffices to find a diagonal matrix $\mathbf{U}$ such that $0\leq U_{i,i}$, $\sum_i U_i \leq 1$ and $\sum_{i} \left(\widetilde{\Lambda}_{i,i}-U_{i,i}\right)^2$ is minimized.

\end{document}